\documentclass[12]{article}
\usepackage{fancyhdr}
\usepackage[colorlinks,citecolor=blue,urlcolor=blue,linkcolor=blue,bookmarks=false]{hyperref}
\usepackage{amsfonts,epsfig,graphicx}
\usepackage{afterpage}
\usepackage{comment}
\usepackage{amsmath,amssymb,amsthm} 
\usepackage{fullpage}
\usepackage[T1]{fontenc} 
\usepackage{epsf} 
\usepackage{graphics} 
\usepackage{amsfonts,amsmath}
\usepackage{mathtools}
\usepackage{enumerate}
\usepackage[compact]{titlesec}
\usepackage{natbib}
\bibpunct{(}{)}{;}{a}{}{,}

\usepackage{psfrag,xspace}
\usepackage{color,etoolbox}
\usepackage{xcolor}
\usepackage{subcaption} 
\usepackage{dsfont}
\usepackage[page]{appendix}
\usepackage{soul} 

\usepackage[a4paper, top=.95in, left= .95in, right= .95in, bottom= .95in]{geometry} 

\titleformat{\section}
  {\normalfont\fontsize{11}{15}\bfseries}{\thesection}{1em}{}

\titleformat{\subsection}
  {\normalfont\fontsize{11}{15}\bfseries}{\thesubsection}{1em}{}

\def\ind{\perp\!\!\!\perp}

\newcommand{\Pb}{\mathbb{P}}
\newcommand{\Pn}{\mathbb{P}_n}

\newcommand{\E}{\mathbb{E}}

\newcommand{\PP}{{\mathbb{P}}}

\DeclareSymbolFont{bbold}{U}{bbold}{m}{n}
\DeclareSymbolFontAlphabet{\mathbbold}{bbold}

\newtheorem{theorem}{Theorem}
\newtheorem{lemma}{Lemma}

\newtheorem{proposition}{Proposition}

\newtheorem{assumption}{Assumption}
\newtheorem{remark}{Remark}
\theoremstyle{definition}

\theoremstyle{remark}

\def\spacingset#1{\renewcommand{\baselinestretch}%
{#1}\small\normalsize} \spacingset{1}

\begin{document}

\pagestyle{plain}

\newcommand{\blind}{0}

\newcommand{\tit}{Local Effects of Continuous Instruments without Positivity}

\if0\blind

{\title{\tit\thanks{Research in this article was supported by the Patient-Centered Outcomes Research Institute (PCORI Awards ME-2021C1-22355) and by the National Library of Medicine, \#1R01LM013361-01A1. All statements in this report, including its findings and conclusions, are solely those of the authors and do not necessarily represent the views of PCORI or its Methodology Committee.  The dataset used for this study was purchased with a grant from the Society of American Gastrointestinal and Endoscopic Surgeons. Although the AMA Physician Masterfile data is the source of the raw physician data, the tables and tabulations were prepared by the authors and do not reflect the work of the AMA. The Pennsylvania Health Cost Containment Council (PHC4) is an independent state agency responsible for addressing the problems of escalating health costs, ensuring the quality of health care, and increasing access to health care for all citizens. While PHC4 has provided data for this study, PHC4 specifically disclaims responsibility for any analyses, interpretations or conclusions. Some of the data used to produce this publication was purchased from or provided by the New York State Department of Health (NYSDOH) Statewide Planning and Research Cooperative System (SPARCS). However, the conclusions derived, and views expressed herein are those of the author(s) and do not reflect the conclusions or views of NYSDOH. NYSDOH, its employees, officers, and agents make no representation, warranty or guarantee as to the accuracy, completeness, currency, or suitability of the information provided here. This publication was derived, in part, from a limited data set supplied by the Florida Agency for Health Care Administration (AHCA) which specifically disclaims responsibility for any analysis, interpretations, or conclusions that may be created as a result of the limited data set.}}
\author{Prabrisha Rakshit\thanks{University of Pennsylvania, United States, Email: prabrisha.rakshit@pennmedicine.upenn.edu}
\and Alexander W. Levis\thanks{Carnegie Mellon University, United States, Email: alevis@andrew.cmu.edu}
\and Luke Keele\thanks{Associate Professor, University of Pennsylvania, United States, Email: luke.keele@gmail.com, corresponding author}
}

\date{\today}

\maketitle
}\fi

\if1\blind
\title{\bf \tit}
\maketitle
\fi

\begin{abstract}
Instrumental variables are a popular study design for the estimation of treatment effects in the presence of unobserved confounders. In the canonical instrumental variables design, the instrument is a binary variable. In many settings, however, the instrument is continuous. Standard estimation methods can be applied with continuous instruments, but they require strong assumptions. While recent work has introduced more flexible estimation approaches, these methods require a positivity assumption that is implausible in many applications. We derive a novel family of causal estimands using stochastic dynamic interventions that allows a range of intervention distributions that are continuous with respect to the observed distribution of the instrument. These estimands focus on a specific local effect but do not require a positivity assumption. Next, we develop doubly robust estimators for these estimands that allow for estimation of the nuisance functions via nonparametric estimators. We use empirical process theory and sample splitting to derive asymptotic properties of the proposed estimators under weak conditions. In addition, we derive methods for profiling the principal strata as well as a method of sensitivity analysis. We evaluate our methods via simulation and demonstrate their feasibility using an application on the effectiveness of surgery for specific emergency conditions.
\end{abstract}

\noindent%

\thispagestyle{empty}

\clearpage

\spacingset{2}

\section{Introduction}

Bias from unmeasured confounding is a primary concern in studies designed to estimate causal effects. Randomization of treatment assignment is perhaps the most well-known strategy to handle unmeasured confounding. Another strategy is the use of an instrumental variable. An instrumental variable (IV), or instrument, is a measured factor that is as good as randomized, is associated with the treatment of interest, but affects outcomes only through its impact on treatment assignment \citep{Angrist:1996}. More concretely for a variable to be a valid instrument it must meet three conditions: (1) it must be associated with the treatment; (2) it must be unconfounded; and (3) it cannot have a direct effect on the outcome \citep{Angrist:1996}. Under these conditions, an IV can provide a consistent estimate of a causal effect even in the presence of unobserved confounding between the treatment and the outcome. See \citet{Baiocchi:2014} and \citet{imbens2014instrumental} for detailed reviews of the IV design.

In the canonical IV design, the instrument is a binary variable. However, in many applications the IV is a multi-valued (discrete or continuous) measure. For example, in comparative effectiveness research, analysts often use a physician's preference for a treatment as an IV for a patient's actual exposure to a treatment \citep{brookhart2006evaluating,keeleegsiv2018}. In applications of this type, the IV may be measured as the percentage of times a surgeon opted to operate for a pre-defined set of medical conditions. This IV, referred to as tendency to operate (TTO), is clearly multi-valued unlike in the canonical IV design. Classical IV estimation methods such as two-stage least squares can be employed with multi-valued IVs such as TTO, but impose a number of restrictions including parametric assumptions for identification, constant treatment effects, and correctly specified outcome models \citep{okui2012doubly}.

More recently, a number of authors have developed estimation methods for the IV design that allow for continuous instruments, incorporate doubly robust adjustment, and allow for heterogenous responses to treatment~\citep{robins2001comment,kennedy2019robust,tan2010marginal,okui2012doubly,vanderlaan2003unified}. Some of this research has targeted the average treatment effect, with identification depending on structural models or other restrictions on treatment effect heterogeneity~\citep{tan2010marginal}. The difficulty with this approach is that it requires some parametric structure for the underlying causal model, for which \emph{a priori} information is generally absent. In practice, misspecification of such models may incur large biases. Alternatively, analysts have instead invoked what is known as the monotonicity assumption \citep{imbens1994}. Under monotonicity, we rule out the possibility that any units respond in opposition to the encouragement provided by the instrument. The advantage of monotonicity is that it is plausible in many applications and allows for nonparametric identification of causal effects among the sub-population of units that respond to the encouragement from the instrument---often referred to as the compliers. Generally, the monotonicity assumption is articulated for binary instruments \citep{imbens1994,tan2006,ogburn2015doubly,takatsu2023doubly}. However, research focused on the estimation of local IV (LIV) curves articulated a version of the monotonicity assumption for the units who comply with instrument encouragement at a given threshold value \citep{heckman1997instrumental,heckman1999local,heckman2005structural,glickman2000derivation,vytlacil2002independence,kennedy2019robust,mauro2020instrumental}. The LIV approach uses a latent index or selection framework to achieve nonparametric identification with a continuous instrument. Early LIV estimators relied on restrictive parametric models and targeted estimands that were fully conditional on all measured covariates
\citep{basu2007use,carneiro2011estimating}. More recent work has focused on developing doubly-robust machine learning based estimation methods for LIV curves \citep{kennedy2019robust,mauro2020instrumental}.  These estimation methods use a semiparametric doubly robust approach that allows for flexible estimation of nuisance functions while possibly attaining parametric rates of convergence for the target estimand. 


While this recent research has made considerable progress toward flexible estimation methods for LIV curves, there remains one key challenge. Specifically, these methods invoke a positivity assumption. In settings with binary instruments, positivity requires that the probability of assignment to the instrument is bounded away from zero and one for all subjects. With multi-valued IVs, positivity requires that no subset of possible IV values are ruled out based on a subject's measured covariates, i.e., the conditional density/mass function of the IV is positive everywhere, which is often unrealistic. For example with TTO, positivity will be violated if each patient in the study population does not have some probability of receiving care from physicians with a range of preferences. That is, if a specific hospital is only served by surgeons with low TTO values, positivity would be violated for patients at that hospital since there are no high TTO surgeons to which they can be assigned. More specifically, positivity requires that for every patient there must be some probability of being able to receive care from every physician within the support of the TTO distribution. Another common continuous IV is distance to a medical facility or college. Positivity is also an unrealistic assumption for this kind of IV, as it would require that for every unit there must be some probability of being every distance within the support of the IV---the maximum distance a unit will travel will depend heavily on where they live. 

To develop LIV methods that do not invoke a positivity assumption, we apply the stochastic dynamic intervention (SDI) framework to the IV design \citep{kennedy2019nonparametric,mauro2020instrumental}. Specifically, we describe a new class of causal effects for continuous IV that can be estimated even if positivity fails. That is, instead of requiring that it be feasible to assign all IV values to every unit, we will consider a range of intervention distributions absolutely continuous with respect to the observed distribution of the IV. Causal effects under these interventions are highly relevant, and can include varying strengths of nudges (upward or downward) to the observed IV value, without considering any values that would be impossible in practice. In addition, we derive robust and efficient estimators, and use empirical process theory and sample splitting to derive their asymptotic properties under weak conditions. Our estimation methods allow us to use highly flexible, nonparametric estimators for the nuisance functions without first order bias. We also include a doubly-robust method of cross validation for selection of hyperparameters, which is critical to avoid overfitting of the LIV curve. Next, we develop methods for profiling of principal strata and sensitivity analysis with respect to the monotonicity assumption that are consistent with the SDI framework. We evaluate our methods via simulation and demonstrate their feasibility using TTO as an IV for receipt of surgery for patients with cholecystitis.


\section{Preliminaries}

\subsection{Notation}

Consider the standard instrumental variable setup, in which the observed data are $n$ iid copies of
$O = (\boldsymbol{X},Z, A, Y) \sim \Pb$, where $\boldsymbol{X} \in \mathcal{X} \subseteq \mathbb{R}^d$ is a vector of covariates, $Z \in \mathbb{R}$ is a continuous instrument, $A \in \{0,1\}$ is binary exposure variable, and $Y \in \mathbb{R}$ is a real-valued outcome of interest. We let $A^z$ and $Y^z$ denote the counterfactual exposure and outcome values, had the instrument been set to $Z = z$. Similarly, we also define $Y^a$ and $Y^{z,a}$ to be the potential outcomes under an intervention that sets $A = a$, and an intervention that sets both $Z = z$ and $A = a$, respectively. We assume that we observe an independent and identically distributed sample $(O_1, ..., O_n)$ from an unknown data-generating distribution $P_0$. We denote by $\mathbb{P}_n$ its corresponding empirical distribution. We index objects by $P$ when they depend on a generic distribution, and we use the subscript $0$ as short-hand for $P_0$. For example, we denote the expectation under $P_0$ by $\E_0$. For a distribution $P$ and any $P$-integrable function $\eta$, we define $P\eta = \int \eta \, dP$. For instance, $\mathbb{P}_n \eta$ represents $n^{-1}\sum_{i=1}^n \eta(O_i)$. Occasionally, we slightly abuse the notation and treat the random variable as an identity mapping to itself, such that, $\mathbb{P}_n Y = n^{-1}\sum_{i=1}^n Y_i$.  Next, we define three nuisance functions: the outcome model $\mu_{P}(x,z) := \mathbb{E}_P(Y \mid X=x, Z = z)$, the model for treatment $,\lambda_{P}(x,z) := \mathbb{E}_P(A \mid X=x, Z = z)$, and the instrument propensity score $\pi_{P}(x,z) := P(Z = z \mid X=x)$. 

\subsection{Assumptions}

Next, we define the key set of assumptions that we use to identify the causal effect of interest.  First, we outline three core IV assumptions that are common to all IV designs.

\label{sec:assump}
\begin{assumption}[Consistency] \label{ass:consistency}
    $A^Z = A$, and $Y^{Z, A} = Y^A = Y$, almost surely.
\end{assumption}

\begin{assumption}[Unconfoundedness] \label{ass:UC}
    $Z \ind (A^z, Y^z) \mid \boldsymbol{X}$.
\end{assumption}

\begin{assumption}[Exclusion Restriction] \label{ass:ER}
    $Y^{z,a} \equiv Y^a$, with probability 1.
\end{assumption}

In words, Assumption~\ref{ass:consistency} asserts that interventions on $Z$ and $A$ are well-defined, and that there is no interference between subjects. Assumption~\ref{ass:UC} asserts that the effect of $Z$ on $A$ and $Y$ is unconfounded, given measured covariates $\boldsymbol{X}$. Assumption~\ref{ass:ER} asserts that the effect of $Z$ on $Y$ acts entirely through its effect on $A$, i.e., $Z$ has no direct effect on $Y$. 

We also invoke a monotonicity assumption that has appeared in the LIV literature (see e.g., \citet{kennedy2019robust}), relevant specifically for continuous instruments.
\begin{assumption}[Monotonicity] \label{ass:mono}
   If $z' > z$ then $A^{z'} \geq A^z$ with probability 1.
\end{assumption}
\noindent We interpret this monotonicity assumption as stipulating that higher values of the instrument can either encourage otherwise unexposed units to be exposed to treatment or have no effect at all. Higher instrument values, thus, cannot discourage treatment exposure relative to lower values. Under this monotonicity assumption, we can still classify units into subpopulations of never-takers, always-takers, and compliers, but there are many possible classifications that are possible. For instance, we can define a complier at $Z = z$ relative to a shift $\delta$ to be a unit for which $A(z) = 1$ but $A(z - \delta) = 0$ for some $\delta > 0$. This version of monotonicity is also employed in \citet{glickman2000derivation} and \citet{vytlacil2002independence}. In these works, it was shown that this continuous version of the monotonicity assumption can be equivalently expressed as the following latent threshold model.
\begin{assumption}[Latent Threshold] \label{ass:thresh}
  $A^z = 1(z \geq T)$ for an unobserved random threshold $T$. 
\end{assumption}
\noindent Here, each complier has some instrument value---the threshold $T$---above which they are exposed to the treatment. For lower values of the instrument, that unit is exposed to control. Large values of $T$ thus imply that it requires higher values of the IV to encourage treatment exposure, and that units are inherently less willing to take the treatment. It is clear that Assumption~\ref{ass:thresh} implies Assumption~\ref{ass:mono}, but to see the reverse implication, observe that Assumption~\ref{ass:mono} yields the following definition of $T$:
\[T = \inf\{z \in \mathbb{R} : A^z = 1\},\]
where $\inf{\emptyset} \equiv \infty$. In words, $T = \infty$ when $A^z = 0$ for all $z$ (``never takers''), $T = -\infty$ when $A^z = 1$ for all $z$ (``always takers''), and otherwise $T \in \mathbb{R}$ represents the smallest value such that $A^z = 1$, i.e., $A^T = 1$ but $A^{T - \delta} = 0$ for all $\delta > 0$.




\section{Formulation of Causal Effects}
A frequently targeted causal estimand in IV analyses is the so-called local average treatment effect (LATE) parameter. The usual LATE for a binary IV, $Z$, measures the effect of treatment $A$ on individuals who receive treatment if assigned $Z = 1$ but not if $Z = 0$. In cases where $Z$ is not binary, researchers may aim to estimate a variation of the LATE, where compliers correspond to those who undergo treatment at a particular $Z=z$ but not at another 
$Z=z^{\prime}$.
\begin{equation}
    \psi_{LATE} = \E\left(Y^1 - Y^0 \mid A^z > A^{z^{\prime}}\right)
\end{equation}
This parameter requires stronger identifying assumptions than the usual LATE in the case of continuous IVs. An alternative strategy for handling non-binary $Z$ involves adopting the classical IV approach, often employing two-stage least squares or a similar method, which estimates the LATE within a regression framework. While this parametric approach is widely used, it heavily relies on the assumption that the researcher can explicitly specify the relationships among potentially high-dimensional covariates, the instrument, treatment, and outcome variables. Next, we develop a dynamic intervention approach (coupled with influence-function-based estimation) that allows us to circumvent these strong parametric assumptions and still accurately estimate relevant causal effects. 

\subsection{Intervention distributions}
To define causal effects, we need to specify one or a set of interventions of interest. In this context, we leverage ideas from the literature on stochastic interventions \citep{kennedy2019nonparametric, diaz2020causal} and modified treatment policies \citep{diaz2012population, haneuse2013, young2014}. Specifically, we consider interventions that modify the distribution of the instrument, wherein we substitute the observational IV density (or ``propensity score''), $\pi(z \mid X) \coloneqq p(z \mid X)$, with a tilted distribution described below. These interventions are dynamic, meaning they depend on the patient's covariate history, since they are tailored to individual patient characteristics via the IV propensity score.

Concretely, we consider interventions under which the observed density for the propensity score, $\pi$, is replaced by an alternative density $q$. In extant literature, such interventions are described purely stochastically. Here, we introduce a novel modified treatment policy formulation to be described in Section~\ref{sec:proposed}. The proposed formulation is very general, and identification will be possible in the absence of a positivity assumption, so long as $q \ll \pi$, i.e., the intervention density $q$ is absolutely continuous with respect to $\pi$. With that said, we focus our analysis on a class of exponentially tilted densities $q_{\delta}$ \citep{diaz2020causal, schindl2024}, indexed by $\delta \in \mathbb{R}$, whose cumulative distribution functions are given by
\begin{equation}
    Q_{\delta}(z \mid X) = \frac{\int_{-\infty}^{z}e^{\delta u}\pi(u \mid X)du}{\int_{-\infty}^{\infty}e^{\delta u}\pi(u \mid X)du}.
    \label{eq: tilted intervention}
\end{equation}
The tilt parameter $\delta$, defined by the user, controls the degree to which the IV propensity scores $\pi(z \mid X)$ deviate from their intervened values $q_{\delta}(z \mid X)$. As such, investigators can use different values of $\delta$ to define different counterfactual scenarios. For example, suppose $Z \in [0,1]$ is a bounded instrument, such that $\pi(z \mid X) > 0$ for all $z \in [0,1]$, almost surely. Then for $\delta = -\infty$, the intervention distribution would assign all units the instrument value $Z = 0$, irrespective of their $X$ values. Similarly, if $\delta = \infty$, it would always assign the value $Z = 1$. Selecting intermediate positive (negative) values of $\delta$ define scenarios where the distribution of $Z$ is shifted up (down) under the intervention. 

\subsection{Proposed estimands}\label{sec:proposed}
Let $A^{Q_{\delta}}$ be the potential treatment had the instrument been randomly sampled from $Q_{\delta}$, given covariates $X$. For $\delta > 0$, we would like to define the causal quantity of interest to be the average treatment effect for those units who would take treatment if the distribution of the instrument is $Q_{\delta}$ but not under the observed distribution (e.g., among those with $A^{Q_{\delta}} > A$). However, such a notion of ``compliers'' is difficult to interpret, in that even though $Q_{\delta}$ may be stochastically greater than $\Pi$, whether or not $A^{Q_{\delta}} > A$ will depend on the particular stochastic value that the instrument happens to take under the intervention.

To remedy this issue, we introduce more a deterministic notion of the modified instrument value: for any intervention cumulative distribution function (CDF) $Q(\, \cdot \mid X)$, define
\begin{equation} \label{eq:mod-inst}
Z_{Q} \coloneqq Q^{-1}\left(\Pi(Z\mid X) \mid X\right),
\end{equation}
where $\Pi(z \mid X) \equiv P[Z \leq z \mid X]$ is the conditional CDF of the observed instrument $Z$, and $Q^{-1}$ denotes the usual quantile function corresponding to $Q$. In words, $Z_Q$ is taken to be a function of $Z$ itself, such that it follows the target intervention distribution. We note that a variant on this transformation also appeared in \citet{diaz2013assessing}. As the next result shows, the transformed instrument $Z_Q$ is rank-preserving, such that compliers---with respect to the counterfactual treatments $A^{Z_{Q}}$---are well-defined.

\begin{proposition}\label{prop:tilt-LATE}
    For any target distribution $Q$, $Z_{Q}$ defined in~\eqref{eq:mod-inst} has distribution $Q(\, \cdot \mid X)$ given $X$. Moreover, under Assumption~\ref{ass:mono}, if $Q$ stochastically dominates $\Pi$ (i.e., $\Pi(z\mid X) \geq Q(z \mid X)$ for all $z \in \mathbb{R}$), then $P[A^{Z_{Q}} \geq A] = 1$. Similarly, if $Q$ is stochastically dominated by $\Pi$, then $P[A^{Z_{Q}} \leq A] = 1$.
\end{proposition}

Focusing on the tilted distributions in~\eqref{eq: tilted intervention}, we define $Z_\delta \equiv Z_{Q_{\delta}}$, for some arbitrary value $\delta \in \mathbb{R}$. The next result shows that $Q_{\delta}$ satisfies the desired stochastic dominance property.

\begin{proposition}\label{prop:dominance}
    For $\delta > 0$, $Q_\delta$ stochastically dominates $\Pi$; for $\delta < 0$, $Q_{\delta}$ is stochastically dominated by $\Pi$.
\end{proposition}

\noindent See the Appendix~\ref{sec:proofs} for proofs of both propositions. With Propositions~\ref{prop:tilt-LATE} and \ref{prop:dominance} in hand, we can define the target estimands of interest. Namely, for $\delta > 0$, we define the $\delta$-tilted LATE to be
\[
\psi(\delta) \coloneqq \mathbb{E}(Y^1 - Y^0 \mid A^{Z_\delta} > A),
\]
and similarly for $\delta < 0$, the $\delta$-tilted LATE is defined to be $\mathbb{E}(Y^1 - Y^0 \mid A^{Z_\delta} < A)$. The result of Proposition~\ref{prop:dominance} not only aids in technically establishing the monotonicity property for defining and identifying (see Section~\ref{sec:main}) the $\delta$-tilted LATE, but also assists in interpretation of this estimand. Namely, for $\delta > 0$, $\psi(\delta)$ represents the average treatment effect for individuals who would opt for treatment if the observed $Z$ were replaced by $Z_\delta$ (whose distribution is $Q_{\delta}$), but are not treated under $Z$. The fact that $Q_{\delta}(\, \cdot\mid X)$ stochastically dominates $\Pi(\, \cdot \mid X)$ implies that the former distribution has a higher probability of leading individuals to choose treatment for a given set of observed characteristics $X$. Consequently, here the complier population can be defined as those who receive treatment when the instrumental variable is `stochastically' elevated by an amount defined by $\delta$ (i.e., more likely to lead to treatment) but do not receive treatment at the instrument's natural level.

\begin{remark}
    In the main paper, we analyze the effect $\psi(\delta) := \E(Y^1 - Y^0 \mid A^{Z_{\delta}} > A)$ for $\delta > 0$. These results can be generalized for any arbitrary $\delta_1 > \delta_2$ with the parameter $\psi(\delta_1, \delta_2) := {\color{black}\E(Y^1 - Y^0 \mid A^{Z_{\delta_1}} > A^{Z_{\delta_2}})}$---by construction, $\psi(\delta) \equiv \psi(\delta, 0)$ since $Z \equiv Z_0$. Refer to the Appendix \ref{app:d1d2} for more precise details on defining and identifying $\psi(\delta_1, \delta_2)$. All other results discussed for $\psi(\delta)$ can be generalized to $\psi(\delta_1, \delta_2)$.
    The discussion is also applicable for $\delta < 0$, e.g., one can consider $\psi(0, \delta) = \E(Y^1 - Y^0 \mid A > A^{Z_{\delta}})$ when $\delta < 0$.
\end{remark}

\section{Main Results}
\label{sec:main}

\subsection{Identification}

Next, we consider identification of the $\delta$-tilted LATE. First, we provide a brief review of identification of the standard LATE estimand. The usual LATE estimand with binary $Z$ is defined and identified as:
\begin{equation}
    \psi_{LATE} = \E\left(Y^1 - Y^0 \mid A^1 > A^0\right) = \frac{\E\left(\mu(1,X) - \mu(0,X)\right)}{\E\left(\lambda(1,X) - \lambda(0,X)\right)}.
\end{equation}
As we noted earlier, identification of this LATE estimand requires a positivity assumption. The positivity assumption asserts that each subject has a positive probability of receiving all levels of $z$, given their covariates. In the case of binary $Z$, this implies that both $Z = 0$ and $Z=1$ must have a positive probability. However, when $Z$ is continuous or multi-valued rather than binary, this assumption suggests that each subject must have a positive conditional probability (or density) of receiving each value of $z \in \text{Supp}(Z)$, given their covariates, which is generally implausible. 

Fortunately, we can identify the $\delta$-tilted LATE without a positivity assumption. As usual for local effects based on instrumental variables, though, we do require a relevance assumption.

\begin{assumption}[Relevance] \label{ass:relevance}
    $\iint \mathbb{E}(A \mid Z=u, X=x) \, dQ_\delta(u \mid X=x) \, d\mathbb{P}(x) \neq \mathbb{E}(A)$.
\end{assumption}

In words, Assumption~\ref{ass:relevance} requires that the expected number of treated individuals under the proposed intervention on $Z$ differs from the expected number under no intervention, i.e., that the intervention on $Z$ has \textit{some} marginal effect on $A$. With this assumption in hand, we obtain the following identification result.
\begin{proposition}
    Under Assumptions~\ref{ass:consistency}--\ref{ass:mono}, and \ref{ass:relevance}, the $\delta$-tilted LATE is identified as 
\begin{equation}
    \psi(\delta) = \frac{\iint \mathbb{E}(Y \mid Z=u, X=x) \, dQ_\delta(u \mid X=x) \,d\mathbb{P}(x)-\mathbb{E}(Y)}{\iint \mathbb{E}(A \mid Z=u, X=x) \, dQ_\delta(u \mid X=x) \, d\mathbb{P}(x)-\mathbb{E}(A)}.
\label{eq: tilted LATE}
\end{equation}
\end{proposition}
\noindent This estimand retains the same basic form as the standard LATE estimand for binary instruments. However, instead of obtaining predicted values under $Z=0$ and $Z=1$, we do so under tilted versions of the observed values of $Z=z$ itself. Note that there are no constraints on the instrument propensity score \(\pi(z \mid X)\). This is because changes in \(q_{\delta}\) following \eqref{eq: tilted intervention} preserves the instrument propensity score when it is zero.

\subsection{Estimation}
The next step is to estimate the identified parameter in \eqref{eq: tilted LATE}. One natural method of estimation for $\psi(\delta)$ is the following plug-in estimator:
\begin{equation}\label{eq:plug-in}
\widehat{\psi}_{\delta}^{reg} = \frac{\Pn \left(\widehat{\gamma}_{\delta}^Y(X) - Y\right)}{\Pn \left(\widehat{\gamma}_{\delta}^A(X) - A\right)},
\end{equation}
where $\widehat{\gamma}_{\delta}^Y(X)$ and $\widehat{\gamma}_{\delta}^A(X)$ are estimates of $\gamma^Y_{\delta}(X) = \frac{\E\left(Ye^{\delta Z}\mid X\right)}{\alpha_{\delta}(X)}$ and $\gamma^A_{\delta}(X) = \frac{\E\left(Ae^{\delta Z}\mid X\right)}{\alpha_{\delta}(X)}$ respectively, with $\alpha_{\delta}(X) = \E\left(e^{\delta Z}\mid X\right)$. We refer to these quantities as nuisance functions on which the estimated $\delta$-tilted LATE depends, but are not of direct interest themselves. 
Estimation of the nuisance functions amounts to a set of standard regression problems, for which we might employ parametric regression models. For example, to estimate $\gamma_{\delta}(X)$, one might assume that the regression functions $\E\left(Ye^{\delta Z}\mid X\right)$ and $\E\left(e^{\delta Z}\mid X\right)$ are linear functions of the covariates. Then, for instance, $\widehat{\gamma}_{\delta}^Y(X)$ is computed as the ratio of the predictions from a linear regression of $Ye^{\delta Z}$ on the covariates to those from a linear regression of $e^{\delta Z}$ on the covariates. Taking the ratio of empirical means computed in a similar manner provides the estimate of $\psi(\delta)$. Using parametric models for the nuisance functions, however, may result in bias from model misspecification. Alternatively, nonparametric estimators, such as kernel smoothing or random forests, could be used for these nuisance functions. However, the plug-in estimator~\eqref{eq:plug-in} for $\psi(\delta)$ will typically inherit any smoothing bias present in the nuisance function estimates \citep{kennedy2016semiparametric, kennedy2022semiparametric}. Specifically, we can describe the plug-in style estimator as having first-order bias, since it will inherit any bias present in the estimates of the nuisance functions.  Consequently, the plug-in estimator may not achieve optimal $\sqrt{n}$-rates of convergence.

Recent advances have sought to address these issues. Specifically, influence function (IF) based estimation offers a significant advantage whereby one can avoid many of the unrealistic parametric assumptions necessary for the plug-in method to achieve fast convergence rates. Using the IF, researchers can construct estimators that possess excellent attributes, such as double robustness and general second-order bias. These attributes permit rapid parametric convergence rates even when nuisance functions are estimated at slower rates through flexible and sophisticated machine learning techniques. The role of the efficient IF is especially noteworthy, as its variance matches the semiparametric efficiency bound, serving as a crucial benchmark for optimal estimator construction \citep{bickel1993efficient}. Influence functions are pivotal for characterizing the asymptotic properties of estimators, since any regular asymptotically linear estimator can be expressed as the empirical average of an influence function plus a negligible $o_P(1/\sqrt{n})$ error term. 

Here, we derive the efficient influence function for the tilted LATE parameter $\psi(\delta)$. As outlined above, this derivation yields the efficiency bound for estimating $\psi(\delta)$, and will be exploited to construct a robust estimator of the tilted LATE. 

\begin{theorem}
    Under the causal conditions the efficient influence function of the functional $\psi(\delta)$ is given by
    \begin{equation}
        \varphi_{\delta}(z, \eta, \psi) =   \frac{\Theta(Y ; \delta, 0) - \psi\Theta(A ; \delta, 0)}{\E\left(\gamma_{\delta}^A(X) - A\right)}
        \label{eq: eff IF}
    \end{equation}
    where for any random variable $T, \Theta(T; \delta, 0) := \xi(T; \delta) - \xi(T; 0)$ for 
    $$
    \xi(T; \delta) := T\frac{e^{\delta Z}}{\alpha_{\delta}(X)} + \gamma_{\delta}^T(X)\left(1 - \frac{e^{\delta Z}}{\alpha_{\delta}(X)}\right). 
    $$
    The nuisance parameter vector $\eta$ consists of $\alpha_{\delta}(X) = \E\left(e^{\delta Z}\mid X\right), \gamma^Y_{\delta}(X) = \frac{\E\left(Ye^{\delta Z}\mid X\right)}{\alpha_{\delta}(X)}$ and $\gamma^A_{\delta}(X) = \frac{\E\left(Ae^{\delta Z}\mid X\right)}{\alpha_{\delta}(X)}$.
\label{thm : eff IF}
\end{theorem}
\noindent 
See Appendix~\ref{sec:proofs} for a full proof of Theorem~\ref{thm : eff IF}, along with proofs of all subsequent results. Efficiency theory for the numerator (or equivalently, the denominator) of $\psi(\delta)$ was considered in \citet{diaz2020causal}, in the context of mediation analysis, and in \citet{schindl2024}, in the context of observational causal inference with continuous exposures---our result is consistent with this prior work. Moreover, \citet{schindl2024} develop minimax lower bounds for estimating this numerator functional, establishing a strict dependence on the tilt parameter $\delta$, and we conjecture that an analogous dependence holds for the $\delta$-tilted LATE, $\psi(\delta)$.

An established method for constructing estimators through influence functions involves solving an estimating equation using the estimated influence function as the estimating function. Concretely, we seek to find $\widehat{\psi}$ that solves $\Pn \left\{\varphi_{\delta}\left(z; \widehat{\eta}, \widehat{\psi}\right)\right\}=0$. Building on the formulation of $\varphi_{\delta}$ given in theorem \ref{thm : eff IF}, we propose the influence-function-based estimator as follows:
\begin{equation}
    \widehat{\psi}_{IF}(\delta) = \frac{\Pn \widehat{\Theta}(Y; \delta, 0)}{\Pn \widehat{\Theta}(A; \delta, 0)}
    \label{eq: IF based LATE}
\end{equation}
for $\widehat{\Theta}$ an estimate of $\Theta$. All that remains is to select a nonparametric method to estimate the nuisance parameters. In the next section, we will demonstrate that attaining nonparametric rates for the nuisance parameters may be sufficient for achieving parametric rates for the target causal estimand, $\psi(\delta)$. In practical applications, investigators could also use a SuperLearner \citep{van2007super}, an ensemble learning method, to estimate the parameters $\alpha_{\delta}, \gamma_{\delta}^Y$ and $\gamma_{\delta}^A$.

\subsection{Asymptotic Theory}
This section examines the behavior of our proposed estimator in large samples, demonstrating that it can achieve rapid $\sqrt{n}$-rates of convergence, even when utilizing flexible nonparametric machine learning methods to estimate the nuisance functions.
\begin{theorem}
    Suppose that the nuisance functions in $\eta$ and their estimators are estimated on a separate sample. Further assume that $\|1-\frac{\alpha_{\delta}(X)}{\widehat{\alpha}_{\delta}(X)}\|+\|\widehat{\gamma}_{\delta}^Y(X) - \gamma_{\delta}^Y(X)\|+\|\widehat{\gamma}_{\delta}^A(X) - \gamma_{\delta}^A(X)\| = o_{\Pb}(1)$. Then,
    \begin{equation}
        \widehat{\psi}_{IF}(\delta) - \psi(\delta) = O_{\Pb}\left[1/\sqrt{n} + \left\|1-\frac{\alpha_{\delta}(X)}{\widehat{\alpha}_{\delta}(X)}\right\|\left\{\left\|\widehat{\gamma}_{\delta}^Y(X)-\gamma_{\delta}^Y(X)\right\|+ \left\|\widehat{\gamma}_{\delta}^A(X)-\gamma_{\delta}^A(X)\right\|\right\}\right] 
        \label{eq: convergence}
    \end{equation}
    Moreover, if $\left\|1-\frac{\alpha_{\delta}(X)}{\widehat{\alpha}_{\delta}(X)}\right\|\left\{\left\|\widehat{\gamma}_{\delta}^Y(X)-\gamma_{\delta}^Y(X)\right\|+ \left\|\widehat{\gamma}_{\delta}^A(X)-\gamma_{\delta}^A(X)\right\|\right\} = o_{\Pb}(1/\sqrt{n})$ then
    \begin{equation}
        \sqrt{n}\left(\widehat{\psi}_{IF}(\delta) - \psi(\delta)\right) \overset{d}{\rightarrow} N\{0,\E(\varphi_{\delta} \varphi_{\delta}^{\intercal})\}.
        \label{eq: asymptotic normality}
    \end{equation}
    \label{thm: asymptotic normality}
\end{theorem}
\noindent In Theorem~\ref{thm: asymptotic normality}, we demonstrate that when $\sqrt{n}$-rate convergence is achieved, the asymptotic variance of $\psi(\delta)$ is simply the variance of $\varphi_{\delta}$, making it easy to estimate and enabling straightforward closed-form confidence intervals that are asymptotically valid. Moreover, the estimator $\widehat{\psi}_{IF}(\delta)$ converges at a $\sqrt{n}$-rate and achieves optimal efficiency if the product of the nuisance error rates diminishes faster than $n^{-1/2}$. For instance, this occurs if each nuisance error rate decreases faster than $n^{-1/4}$. This condition does not rely on parametric model assumptions and hence allows for the use of advanced machine learning techniques while maintaining classical inference properties. Additionally, $\widehat{\psi}_{IF}(\delta)$ is doubly robust, meaning it remains consistent if either the regression estimator $\widehat{\alpha}_{\delta}(X)$ or the ratios of regression estimators $\left(\widehat{\gamma}_{\delta}^Y(X), \widehat{\gamma}_{\delta}^A(X)\right)$ are consistently estimated, even if not both. 

While the aforementioned theorem supports pointwise inference at specific $\delta$ values, often a confidence band covering a continuous range of $\delta$ shifts is preferable. This entails creating a $95\%$ confidence band that is uniformly valid across the entire range of $\delta$ values considered. Such bands, known as uniform confidence bands, ensure coverage of the true parameter across all measurements in the range. The subsequent theorem addresses how to construct uniform confidence bands for function-valued parameters.
\begin{theorem}
Under the assumptions of Theorem \ref{thm : eff IF}. Further, assume
\begin{enumerate}[(i)]
    \item $\varphi_{\delta}$ is Lipschitz in $\delta$
    \item $\|\hat{\sigma}(\delta) / \sigma(\delta)\|_{\mathcal{D}}-1=o_p(1)$
    
    \item $\sup _{\delta \in \mathcal{D}}\left\{\left\|\widehat{\gamma}^Y_{\delta}(X) - \gamma^Y_{\delta}(X)\right\|\left\|1-\frac{\alpha_{\delta}(X)}{\widehat{\alpha}_{\delta}(X)}\right\|\right\}=o_p(1 / \sqrt{n})$
    
    \item $\sup _{\delta \in \mathcal{D}}\left\{\left\|\widehat{\gamma}^A_{\delta}(X) - \gamma^A_{\delta}(X)\right\|\left\|1-\frac{\alpha_{\delta}(X)}{\widehat{\alpha}_{\delta}(X)}\right\|\right\}=o_p(1 / \sqrt{n})$
\end{enumerate}
Then,
$$
\sup _{\delta \in \mathcal{D}}\left|\frac{\sqrt{ n} \left\{\hat{\psi}_{\mathrm{IF}}(\delta)-\psi(\delta)\right\}}{\hat{\sigma}(\delta)}-\sqrt{n} \left(\mathbb{P}_n-\mathbb{P}\right) \frac{\varphi_{\delta}(z ; \eta, \psi)}{\sigma(\delta)}\right|=o_p(1 / \sqrt{ } n) .
$$
where $\sigma^2(\delta) = \E(\varphi_{\delta}\varphi_{\delta}^{\intercal})$
\label{thm: uniform normality}
\end{theorem}
\noindent This result suggests that we can establish uniform confidence bands for our estimators spanning a continuous spectrum of $\delta$ values with assumptions only slightly more stringent than those for pointwise inference. Specifically, under these conditions, we can employ the multiplier bootstrap method to construct uniform confidence bands, as in \citet{belloni2015, kennedy2019nonparametric}.

\section{Extensions}

Next, we develop two important extensions. Specifically, we focus on profiling of principal strata and a sensitivity analysis for the monotonicity assumption. These are common estimation problems in the canonical IV design. Here, we demonstrate how to implement these tasks under our continuous IV estimation framework.

\subsection{Profiling}
\label{sec: profiling}

One important component of an IV study---particularly when working with a monotonicity assumption---is profiling. In a profiling analysis, the investigator produces descriptive statistics for the principal strata populations: the compliers, always-takers, and never-takers. These principal strata covariate distributions are compared to the overall patient population profile. When there is treatment effect heterogeneity and there are notable differences in these covariate distributions, it implies that the LATE may deviate from the ATE. Profiling was first proposed in \citet{angrist2008mostly} and extended in \citet{Baiocchi:2014}, \citet{marbach2020profiling}, and \citet{takatsu2023doubly}. Extant profiling methods have been developed for the canonical case where the IV is a binary variable. To our knowledge, no methods yet exist for profiling principal strata when the instrument is continuous. Thus, we next develop profiling methods for continuous instruments. Critically, in our framework the definition of the principal strata will depend on the value of the shift parameter $\delta$. That is, shifting $\delta$ higher may increase or decrease the size of the principal strata. Next, we develop DRML-based techniques for nonparametrically estimating the covariate distributions of compliers, always-takers, and never-takers for analyses where $\psi(\delta)$ is the target causal parameter.

First, we denote by $V$ a subset of baseline covariates of interest, i.e., $V \subseteq X$. Initially, we outline profiling methods for when $V$ is a discrete random variable, and briefly consider continuous $V$ below. Under the assumptions in Section~\ref{sec:assump} (notably, Assumption~\ref{ass:mono}), we can identify the probability of the baseline covariates taking value $V = v_0$ among the compliers (for $\delta > 0$) as follows:
\begin{equation}
    \Pb(V = v_0 \mid A^{Z_{\delta}} > A) = \frac{\E[D_{\psi}I(V = v_0)]}{\E[D_{\psi}]},
    \label{eq: comp profiling}
\end{equation}
where $D_{\psi} = \gamma_{\delta}^A(X) - A$. Also note that $\E(D_{\psi})$ is the denominator of the LATE parameter $\psi(\delta)$ and can be estimated using the uncentered influence function of the denominator i.e., $\{1-\frac{e^{\delta R}}{\alpha_{\delta}(X)}\}(\gamma_{\delta}^A(X) - A)$. Therefore, the influence function based estimator of \eqref{eq: comp profiling} is given by
\begin{equation}
    \widehat{\Pb}(V = v_0 \mid A^{Z_{\delta}} > A) = \frac{\Pn\left[\left(\{1-\frac{e^{\delta R}}{\widehat{\alpha}_{\delta}(X)}\}(\widehat{\gamma}_{\delta}^A(X)-A)\right)I(V = v_0)\right]}{\Pn\left[\left(\{1-\frac{e^{\delta R}}{\widehat{\alpha}_{\delta}(X)}\}(\widehat{\gamma}_{\delta}^A(X)-A)\right)\right]}.
\end{equation}
\noindent It follows by Theorem \ref{thm: asymptotic normality} that we can construct a $100(1-\alpha)\%$ confidence interval for $\Pb(V = v_0 \mid A^{Z_{\delta}} > A)$ as
\begin{equation}
    \left[\widehat{\Pb}(V = v_0 \mid A^{Z_{\delta}} > A) - z_{\alpha/2}\E(\varphi_{\delta,c}\varphi_{\delta,c}^{\intercal})/\sqrt{n}, \, \widehat{\Pb}(V = v_0 \mid A^{Z_{\delta}} > A) + z_{\alpha/2}\E(\varphi_{\delta,c}\varphi_{\delta,c}^{\intercal})/\sqrt{n}\right]
\end{equation}
where $\varphi_{\delta,c}$ is same as $\varphi_{\delta}$ in \eqref{eq: eff IF} with $\Theta(Y; \delta, 0)$ replaced by $\Theta(A; \delta, 0)1(V = v_0)$. The probability of compliance for $\delta < 0$, $\Pb(V = v_0 \mid A^{Z_{\delta}} < A)$, can be addressed in a similar manner and can be shown to yield the same expression. Moreover, in case of continuous $V$, we propose to replace $I(V = v_0)$ with $K_{v_0,h}$ where $ K_{v_0,h}$ is the map $v \mapsto \frac{1}{h}K\left(\frac{v - v_0}{h}\right)$, for arbitrary kernel function $K$ and bandwidth parameter $h > 0$ \citep{takatsu2023doubly}.

An expression similar to \eqref{eq: comp profiling} can be derived for always-takers ($A^{Q_{\delta}} = 1 = A$):
\begin{equation}
\Pb\left(V = v_0 \mid A^{Z_{\delta}} = A = 1\right) = \frac{\E\left(1(V = v_0)A \right)}{\E(A)}.
\label{eq: at profiling}
\end{equation}
Next, the expression for never-takers ($A^{Z_{\delta}} = 0 = A$) is
\begin{equation}
\Pb\left(V = v_0 \mid A^{Z_{\delta}} = A = 0\right) = \frac{\E\left(1(V = v_0)\{1 - \gamma_{\delta}^A(X)\}\right)}{\E\left(1 - \gamma_{\delta}^A(X)\right)}.
\label{eq: nt profiling}
\end{equation}
Notably, the expression for always-takers does not depend on $\delta$, unlike that for compliers and never-takers. The expressions for never-takers and always-takers when \(\delta < 0\) are given in the Appendix~\ref{sec: app profiling}. Construction of confidence intervals for when $V$ is continuous follows a similar logic as to when $V$ is discrete. Finally, in Appendix~\ref{sec: app defier}, we include some brief comments on how, when relaxing the monotonicity assumption (Assumption~\ref{ass:mono}), the observed proportion of compliers constrains the proportion of defiers for our proposed estimand.

\subsection{Sensitivity Analysis}
\label{sec: sensitivity}
As we stipulated in Section~\ref{sec:assump}, we have assumed that monotonicity holds. However, a number of authors have noted that the monotonicity assumption may be less plausible in settings with continuous instruments \citep{swanson2014think,swanson2017challenging,small2017instrumental}. Specifically, monotonicity violations may be more likely when the instrument is not delivered uniformly to all subjects, which is generally the case for continuous instruments. Given this possibility, we extend a method of sensitivity analysis method introduced in \cite{takatsu2023doubly} to the continuous instrument setting. This form of sensitivity analysis is predicated on a result in \citet{imbens1994}, which demonstrated that when monotonicity is violated, the IV estimand can be written as the sum of the standard IV estimand and a term that depends on two, additional parameters.

In our context, when the monotonicity assumption is violated, then the $\delta$-tilted LATE can be decomposed into the functional $\psi(\delta)$ in equation \eqref{eq: tilted LATE} and an additional term that depends on two parameters: $\gamma_1(\delta)$ and $\gamma_2(\delta)$. We provide here an explanation for $\delta > 0$, and note that similar expressions can be derived for $\delta < 0$ (see Appendix \ref{sec: app sensitivity} for details):
\begin{equation}
    \E\left(Y^1 - Y^0 \mid A^{Z_{\delta}} > A\right) = \psi(\delta) + \frac{\gamma_1(\delta)\gamma_2(\delta)}{\E\left(A^{Z_{\delta}} - A\right)},
\end{equation}
where for a pre-specified $\delta, \gamma_1(\delta)$ signifies the proportion of defiers, while $\gamma_2(\delta)$ represents the discrepancy in average treatment effects between defiers and compliers. Notably, these defiers and compliers are defined in relation to the tilted stochastic intervention $Q_{\delta}(z \mid X)$. Specifically, $\gamma_1(\delta)$ and $\gamma_2(\delta)$ are defined as
\begin{equation*}
    \begin{aligned}
        \gamma_1(\delta) & := \Pb\left(A^{Z_{\delta}} < A\right) \\
        \gamma_2(\delta) & := \E\left(Y^1 - Y^0 \mid A^{Z_{\delta}} < A\right) - \E\left(Y^1 - Y^0 \mid A^{Z_{\delta}} > A\right)
    \end{aligned}
\end{equation*}


Notably, larger values of $\gamma_1(\delta)$ and $\gamma_2(\delta)$ indicate a greater disparity between $\psi(\delta)$ and the tilted LATE. To address this, we utilize influence function-based estimators for $\psi(\delta)$ and $\E\left(A^{Z_{\delta}} - A\right)$, and we examine the following quantity for any fixed $(\gamma_1(\delta),\gamma_2(\delta))$: 
\begin{equation}
    \widehat{\xi}_{\delta}(\gamma_1(\delta),\gamma_2(\delta)) \coloneqq \widehat{\psi}_{IF}(\delta) + \frac{\gamma_1(\delta) \gamma_2(\delta)}{\Pb_n\left[\left(\{1-\frac{e^{\delta R}}{\widehat{\alpha}_{\delta}(X)}\}(\widehat{\gamma}_{\delta}^A(X)-A)\right)\right]}
    \label{eq: sensitivity}
\end{equation}
This expression quantifies the estimated difference between the true $\delta$-tilted LATE and the identified parameter $\psi(\delta)$, as a function of $\gamma_1(\delta)$ and $\gamma_2(\delta)$ for a fixed $\delta$. One can then use this expression to visualize the estimated upper or lower bounds on LATE jointly as a function of $\gamma_1(\delta)$ and $\gamma_2(\delta)$. For a fixed $\delta$, any values $\gamma_1(\delta),\gamma_2(\delta)$ that satisfy $\widehat{\xi}_{\delta}(\gamma_1(\delta),\gamma_2(\delta)) = 0$ represent a minimal violation of the monotonicity criterion leading to a change in sign in the estimated LATE. Subsequently, using subject matter expertise, one can then assess whether the values of $\gamma_1(\delta)$ and $\gamma_2(\delta)$ at this boundary appear plausible in a given scientific context---if not, this may be seen to provide evidence of the robustness of results to the monotonicity assumption.

\section{Simulation}
\label{sec: simulation}
Next, we study the properties of our proposed methods using simulation. First, we focus on the general statistical properties of our methods. In the second simulation, we focus on inferential properties. In both studies, we use a common data generating process (DGP) that we describe next.  Specifically, we simulate data from the following set of models:
\begin{equation}
    \begin{aligned}
        (Y^0, X) & \sim N(0,I_5), \\
        Z \mid X, Y^0 & \sim N(\alpha^{\intercal}X, 2), \\
        A & = I(Z \geq Y^0), \\
        Y & = Y_0 + \psi A.
    \end{aligned}
\end{equation}
In the simulations, we set $\alpha = \left(1, 1, -1, -1\right)$, and the causal effect $\psi =2$. Note that this DGP is consistent with all the identifying assumptions we've outlined above.

\subsection{Simulation Study 1}

In the first simulation study, we compare our proposed influence-function-based estimators to a plug-in estimator. The plug-in estimator is roughly equivalent to a two-stage least-squares-type estimator, adapted to estimate the same parameter as the shift estimator. Our influence-function-based estimator $\widehat{\psi}_{IF}$ in \eqref{eq: IF based LATE} offers users the flexibility to choose a variety of methods for estimating the nuisance parameters $\alpha_{\delta}(\cdot), \gamma_{\delta}^Y(\cdot)$ and $\gamma_{\delta}^A(\cdot)$. Here, we use SuperLearner, an ensemble learning technique, to estimate the nuisance functions. Specifically, we use an ensemble of learners that include a generalized additive model, a generalized linear model, and a random forest, among others. This collection of learners allows for flexible estimation of the nuisance functions, mitigating the possibility of bias from using parametric models. We assessed estimator performance via integrated bias and root-mean-squared error:
\begin{equation}
\widehat{\operatorname{bias}}=\frac{1}{I} \sum_{i=1}^I\left|\frac{1}{J} \sum_{j=1}^J \hat{\psi}_j\left(\delta_i\right)-\psi\left(\delta_i\right)\right|, \quad \widehat{\operatorname{RMSE}}=\frac{\sqrt{n}}{I} \sum_{i=1}^I\left[\frac{1}{J} \sum_{j=1}^J\left\{\hat{\psi}_j\left(\delta_i\right)-\psi\left(\delta_i\right)\right\}^2\right]^{1 / 2}
\end{equation}
across 12 equispaced values of $\delta \in [-0.85,0.85]$. In the simulations, we used 500 replications with sample sizes of $500$, $1000$, and $5000$ for each value of $\delta$. 

The results are summarized in Figure~\ref{fig: Ibias} and Figure~\ref{fig: Irmse}.

\begin{figure}[ht!]
    \centering
    \begin{minipage}[b]{0.45\linewidth}
        \centering
        \includegraphics[width=\linewidth]{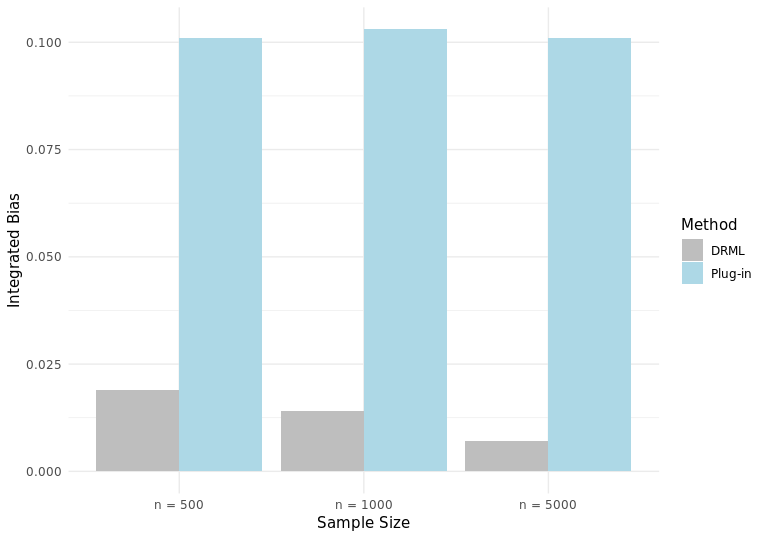}
        \caption{Integrated bias}
        \label{fig: Ibias}
    \end{minipage}
    \hspace{0.05\linewidth}
    \begin{minipage}[b]{0.45\linewidth}
        \centering
        \includegraphics[width=\linewidth]{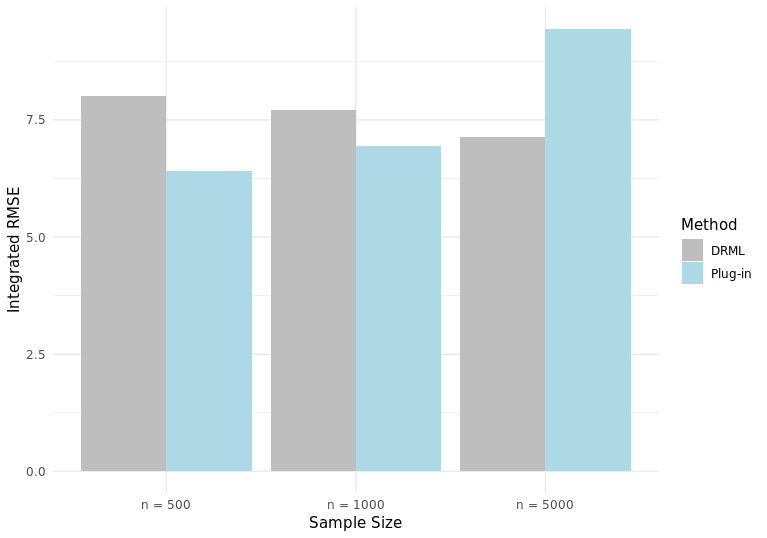}
        \caption{Integrated root mean squared error}
        \label{fig: Irmse}
    \end{minipage}
\end{figure}

When we compare the proposed DRML estimator to the plug-in estimator, we notice a substantial reduction in the bias for all three sample sizes. (see Fig. \ref{fig: Ibias}). In fact, the bias is more than 50\% lower for the DRML estimator. However, this reduction in bias comes at the cost of increased variance. For sample sizes of 500 and 1000, the plug-in estimator outperforms the DRML estimator on the RMSE metric. Once the sample size is 5000, the DRML estimator has a lower RMSE than the plug-in estimator. As such, the DRML estimator does require larger sample sizes to outperform the plug-in estimator for the RMSE metric.

\subsection{Simulation Study 2}

In the second simulation study, we focus on inference directly via the performance of the confidence intervals. Here, we compare confidence intervals based on the asymptotic standard deviation to the uniform confidence bands based on the multiplier bootstrap method. First, we seek to understand how the performance of the confidence intervals changes as $\delta$ changes. In Figure~\ref{fig: All_n}, we plotted both confidence intervals and point estimates for both the DRML estimator and the plug-in estimator for the same three sample sizes in the first simulation study.

In all scenarios, the confidence intervals (CI) for the DRML estimator consistently covered the true value. Moreover, as the sample size increased, the length of the CI decreased. Notably, the DRML estimator exhibited equivalent or superior performance compared to the plug-in estimator. In particular, the plug-in estimator demonstrated erratic behavior as $\delta$ approached $0$, but the DRML estimator did not. For larger sample sizes, this erratic behavior leads to poor coverage.

\begin{center}
    \begin{figure}[ht!]
        \centering
        \includegraphics[scale = 0.6]{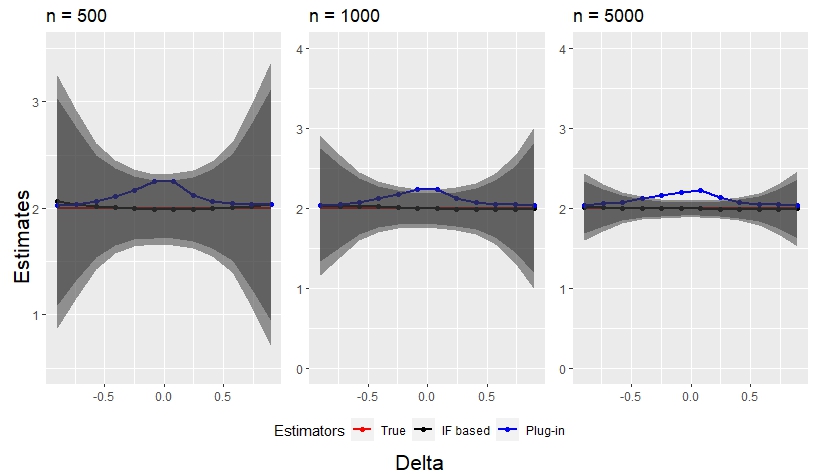}
        \caption{Plug-in (blue) and influence-function-based (black) estimators across $\delta$ values for sample size $n \in \{500, 1000, 5000\}$. CI based on theoretical standard deviation (dark grey) and uniform confidence bands (light grey) are shown.}
        \label{fig: All_n}
    \end{figure}
\end{center}

Next, we further studied the performance of confidence intervals via empirical coverage. In Figure~\ref{fig: coverage-trend}, we plotted 95\% empirical coverage for the DRML-estimator against the $\delta$ values. In Figure~\ref{fig: coverage-trend} each panel contains the results for the three different sample sizes of $n = 500$, $n = 1000$, and $n = 5000$. The confidence intervals display a trend where the coverage drops significantly as $\delta$ approaches zero. Most likely this occurs due to the fact that when $\delta \to 0$ the weight $1 - \frac{e^{\delta Z}}{\E(e^{\delta Z}\mid X)} \to 0$, making the DRML ratio estimate $\widehat{\psi}_{IF}$ approach the ratio $0/0$, which results in poor performance. Towards the more extreme values of $\delta$, the behavior becomes more unpredictable, but generally, when values of $\delta$ are large and positive, the coverage tends to be lower. This could be because larger $\delta$ values mean larger deviations from the observed data. This is analytically supported by the relationship $\frac{q_{\delta}(Z \mid X)}{\pi(Z \mid X)} = \frac{e^{\delta Z}}{\int e^{\delta u}\pi(u \mid X)du} \leq e^{C\delta}$, which indicates that the ratio can increase exponentially with $\delta$. 

\begin{center}
    \begin{figure}[ht!]
        \centering
        \includegraphics[scale = 0.6]{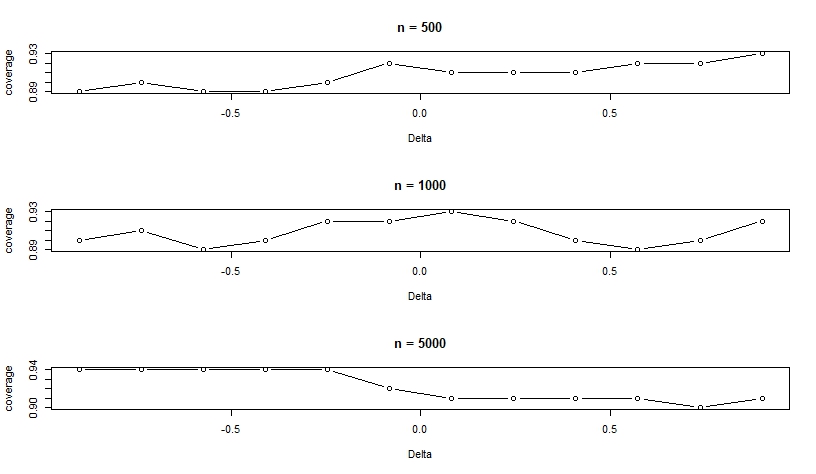}
        \caption{Empirical coverage averaged over 500 rounds of simulations over different values of $\delta$ for DRML based estimator. The top, middle and bottom panels correspond to $n = 500 ,1000, 5000$ respectively.}
        \label{fig: coverage-trend}
    \end{figure}
\end{center}

\section{Application}

Emergency general surgery (EGS) refers to medical emergencies where the injury is internal (e.g., a burst appendix) while trauma care refers to injuries that are external (e.g., a gunshot wound). There are 51 medical conditions that have been designated as EGS conditions that are distinct from trauma conditions. Patient with EGS conditions account for more than 3-4 million hospital admissions per year with more than 800,000 operations in the United States \citep{gale2014public,shafi2013emergency,scott2016use,havens2015excess,ogola2015financial}. Recent research has focused on the comparative effectiveness of operative versus non-operative care for various EGS conditions \citep{hutchings2022effectiveness,kaufman2022operative,moler2022local,grieve2023clinical}.  Assessing the effectiveness of operative care EGS conditions is difficult since allocation of treatment is likely subject to confounding by indication. In addition, many studies of EGS use medical claims data, and these data may fail to completely remove bias given that prognostic factors are incompletely recorded.

To assess the effectiveness of operative care, researchers have used a physician's preference for a specific treatment as an IV \citep{brookhart2007preference,keeleegsiv2018}. In these studies, each physician preference for operative care serves as a ``nudge'' toward a specific mode of care; where the analyst assumes that the physician's preference has no direct effect on patient outcomes. Specifically, these studies have used a surgeon's tendency to operate (TTO) as an instrument to determine whether a patient receives surgery after admission to the emergency department \citep{keeleegsiv2018}. The assignment of surgeons to patients is plausibly as-if randomized, since this study focuses on patients who are receiving emergency care and are unlikely to be able to select their physician. \citet{keeleegsiv2018} first proposed using TTO as an IV. TTO has been validated in a variety of settings using a wide variety of diagnostic tools \citep{keeleegsiv2018,keeleexrest2018,hutchings2022effectiveness,kaufman2022operative,moler2022local,grieve2023clinical}.

In this application, the outcome $Y$ is a binary indicator for an adverse event (a complication or prolonged length of stay). The treatment $A$ is binary, where $A=1$ indicates that a patient received operative care after their emergency admission, while $A = 0$ indicates that a patient received non-operative care. The instrument $Z$ is measured using the percentage of times a surgeon operates when presented with an emergency surgery condition. Extant work has converted this continuous measure into a binary IV by recoding it such that $Z=1$ when a surgeon has TTO value higher than the sample median \citep{keeleegsiv2018}. In other words, an IV value of 1 indicates that a patient has been assigned to a high TTO surgeon and thus are more likely to receive operative care. Here, we re-analyze data from \citet{takatsu2023doubly} who used an IV of this type to study the effectiveness of surgery for cholecystitis patients. In our analysis, we do not use a binary version of the IV. Instead, we use the continuous version of the IV which is recorded as a proportion of the times a surgeon operated. 

 This data is based on data from the American Medical Association (AMA) Physician Masterfile merged with all-payer hospital discharge claims from New York, Florida and Pennsylvania in 2012-2013. The study population includes all patients admitted for emergency or urgent inpatient care. The data includes patient sociodemographic and clinical characteristics, such as indicators for frailty, severe sepsis or septic shock, and 31 comorbidities based on Elixhauser indices \citep{elixhauser1998comorbidity}. Each patient is linked with his or her surgeon through a unique identifier. The primary outcome is the presence of an adverse event after surgery, i.e., death or a prolonged length of stay in the hospital (PLOS). Finally, baseline covariates $\boldsymbol{X}$ include race, sex, age, insurance type, baseline frailty measures, 31 comorbidities, and hospital fixed effects. In our data, there are 6,184 surgeons and the IV takes on 3,278 unique values that range from 0.01 to 1. Finally, there are 123,420 patients. 

We now apply our proposed methods to the question of whether operative care has an effect on the risk of an adverse event. 
First, we outline how to interpret the results in this context. In this application, the target estimand represents the average treatment effect for patients who would receive surgery if the distribution of TTO was set to $Q_{\delta}$, rather than the observed distribution of TTO. That is, we consider the effect of surgery had TTO been stochastically elevated by $\delta$. Compliers are those patients who receive surgery when TTO is stochastically elevated by $\delta$ but do not receive surgery when TTO is at its current level.  Here, we use a range of $\delta$ parameters that produce larger stochastic shifts in the distribution of TTO.  Each of these possible shifts defines a different complier population and leads to estimates for operative care for slightly different local populations. For example, each shift defines a set of patients who receive surgery when their surgeon's TTO value is stochastically increased by $\delta$, but do not receive surgery at the current TTO level. Critically, our estimator relaxes the strong assumption of positivity, which requires that all patients have a positive probability of having a high TTO surgeon.



We use 5 folds and estimated nuisance functions via random forest fits. Figure~\ref{fig:est.plot} contains the results with the uniform confidence bands. The estimated change in the risk of having an adverse event when patients receive operative care ranges from -0.026\% to -0.029\%.  At all levels of $\delta$, the 95\% and uniform confidence bands do not include zero implying that there is a statistically significant effect. The complier population increases from 2.8\% at $\delta_{min}$ to 17\% at $\delta_{max}$. The confidence bands widen slightly as the shifts get larger. One useful feature of the multiplier bootstrap is that we can test a hypothesis of effect homogeneity. To do so, we check to see if we can draw a horizontal line through the uniform confidence bands. We find we cannot reject a homogeneous effect of -0.025.

\begin{figure}[htbp]
  \centering
    \includegraphics[scale=0.6]{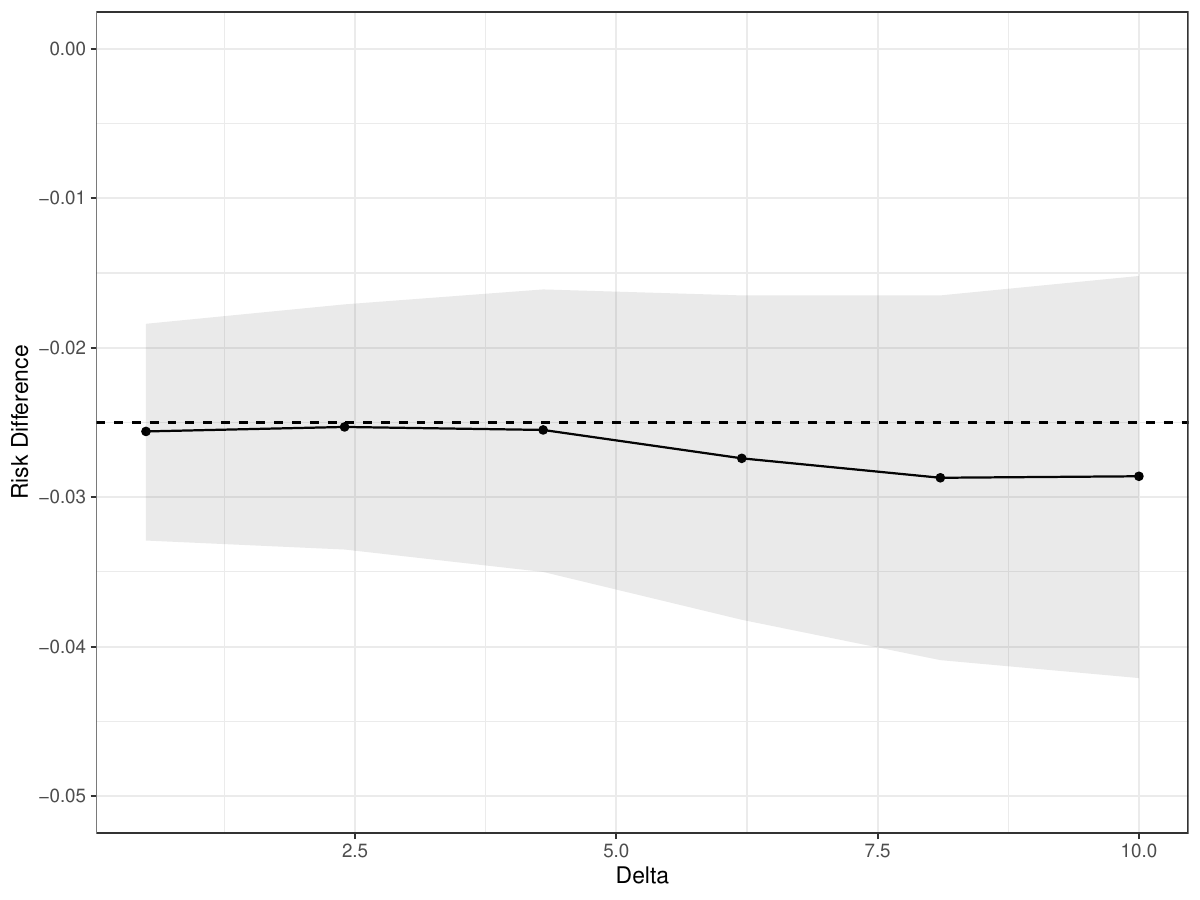}
     \caption{Estimated effect of surgery on the risk of an adverse event by shifting TTO}
  \label{fig:est.plot}
\end{figure}

Next, we profile the principal strata as a function of the values for $\delta$. We examine the results for two key variables: age and sepsis. Figure~\ref{fig:sepsis.plot} contains the results for patients with sepsis. The plot does not include results for always-takers, since the probability of being an always-taker does not depend on $\delta$. We see that for small values of $\delta$, the probability of being a complier or never-taker is nearly identical at 0.12. However, as the magnitude of $\delta$ increases the probability of a patient being a never-taker decreases. 

\begin{figure}[htbp]
  \centering
    \includegraphics[scale=0.4]{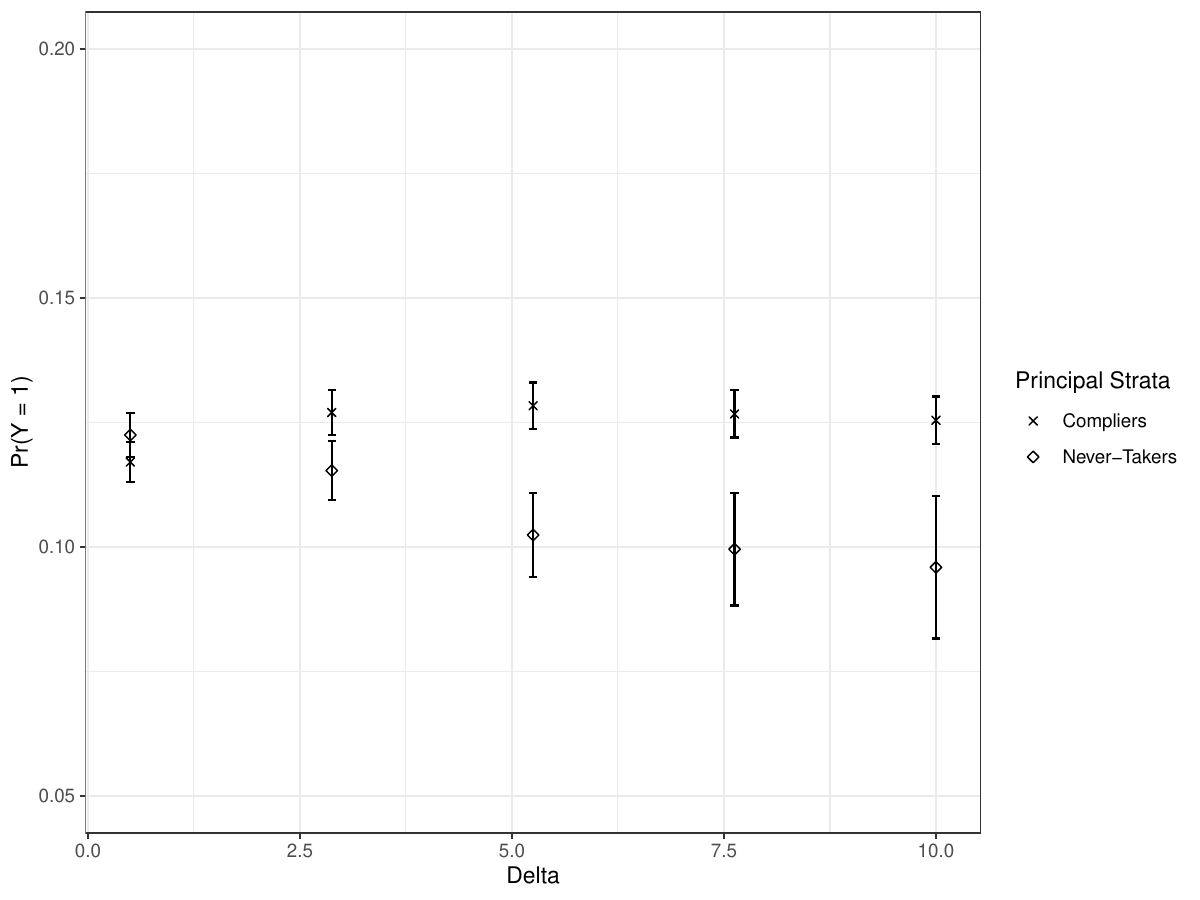}
     \caption{Estimated probability of principal strata for septic patients.}
  \label{fig:sepsis.plot}
\end{figure}

Next, we implemented a profiling analysis for age. In the first analysis, we fixed the value of age at the median value of 62, and then plotted the density value for compliers and never-takers as a function of $\delta$.  Figure~\ref{fig:a1} contains these results. In this case, the estimated densities for compliers and never-takers overlap across the range of $\delta$ values. Next, we fixed the value of $\delta$ and estimated the densities at seven equally spaced values across the support of the age distribution. Results are in Figure~\ref{fig:a2}.  We observe minor differences between compliers and never-takers, but the general pattern is the same. Patients are far more likely to be compliers and never-takers in the middle of the age distribution.

\begin{figure}[ht]
  \centering
  \begin{subfigure}[t]{0.45\textwidth}
    \includegraphics[width=\textwidth]{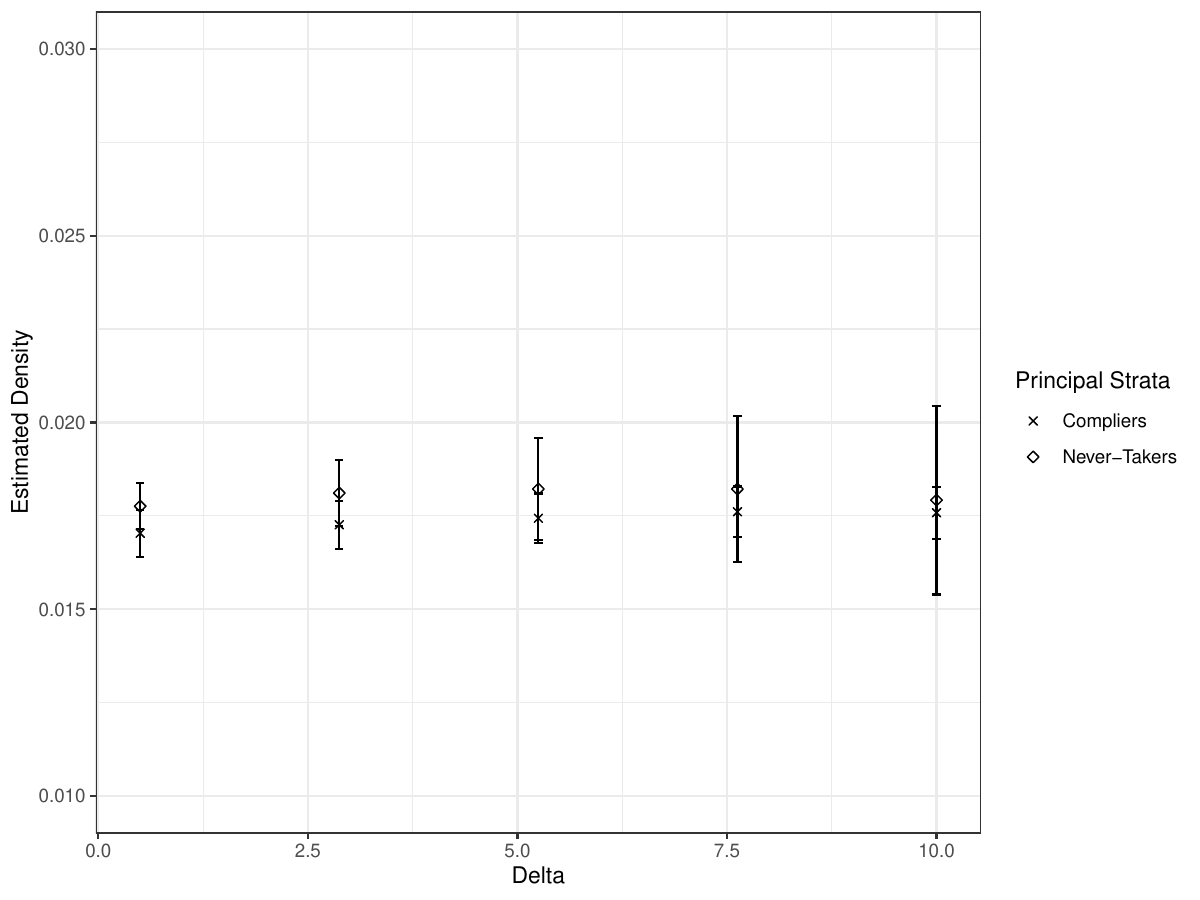}
    \caption{Age density by $\delta$}
    \label{fig:a1}
  \end{subfigure}
  \begin{subfigure}[t]{0.45\textwidth}
    \includegraphics[width=\textwidth]{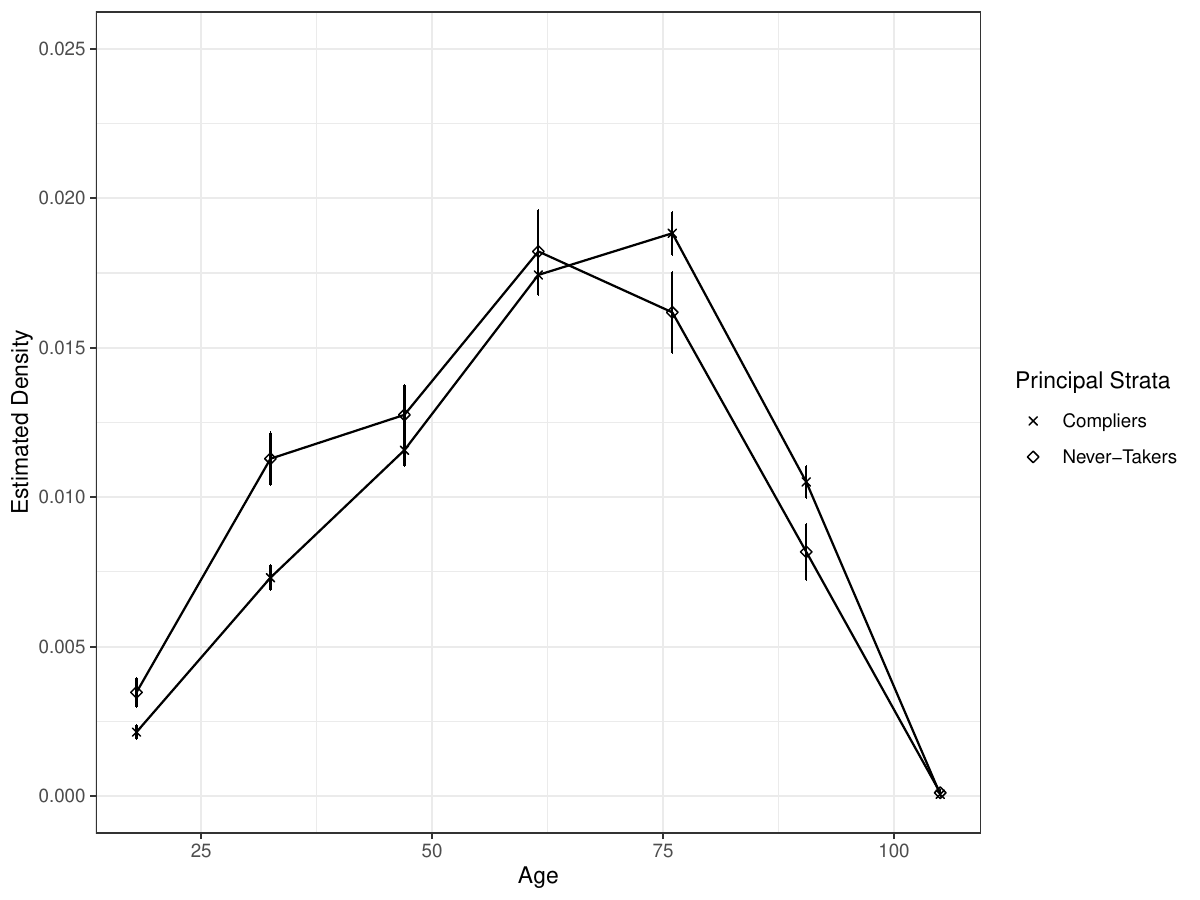}
    \caption{Age density across range of support}
    \label{fig:a2}
  \end{subfigure}
  \caption{Profiling principal strata for age as function of $\delta$ and support of the age distribution.}
  \label{fig:age-profile}
\end{figure} 

Next, we report the results from the sensitivity analysis for the monotoncity assumption. For this analysis, we generate point estimates across a range of values for the two sensitivity parameters: $\gamma_1(\delta)$ and $\gamma_2(\delta)$, Importantly, these parameters depend on $\delta$. Given this dependence, we implemented two separate sensitivity analyses: one for the minimum value of $\delta$ we used and one for the maximum value of $\delta$ we used. The results are in Figure~\ref{fig:mono-sens}. For both analyses, we fixed the ranges of both $\gamma_1(\delta)$ and $\gamma_2(\delta)$. The dotted line in each figure represents the point where the sign of the point estimate changes as a function of the two parameters. Of interest, the estimates are far more robust to violations of the monotonicity assumption for the larger value of $\delta$. 

\begin{figure}[ht]
  \centering
  \begin{subfigure}[t]{0.45\textwidth}
    \includegraphics[width=\textwidth]{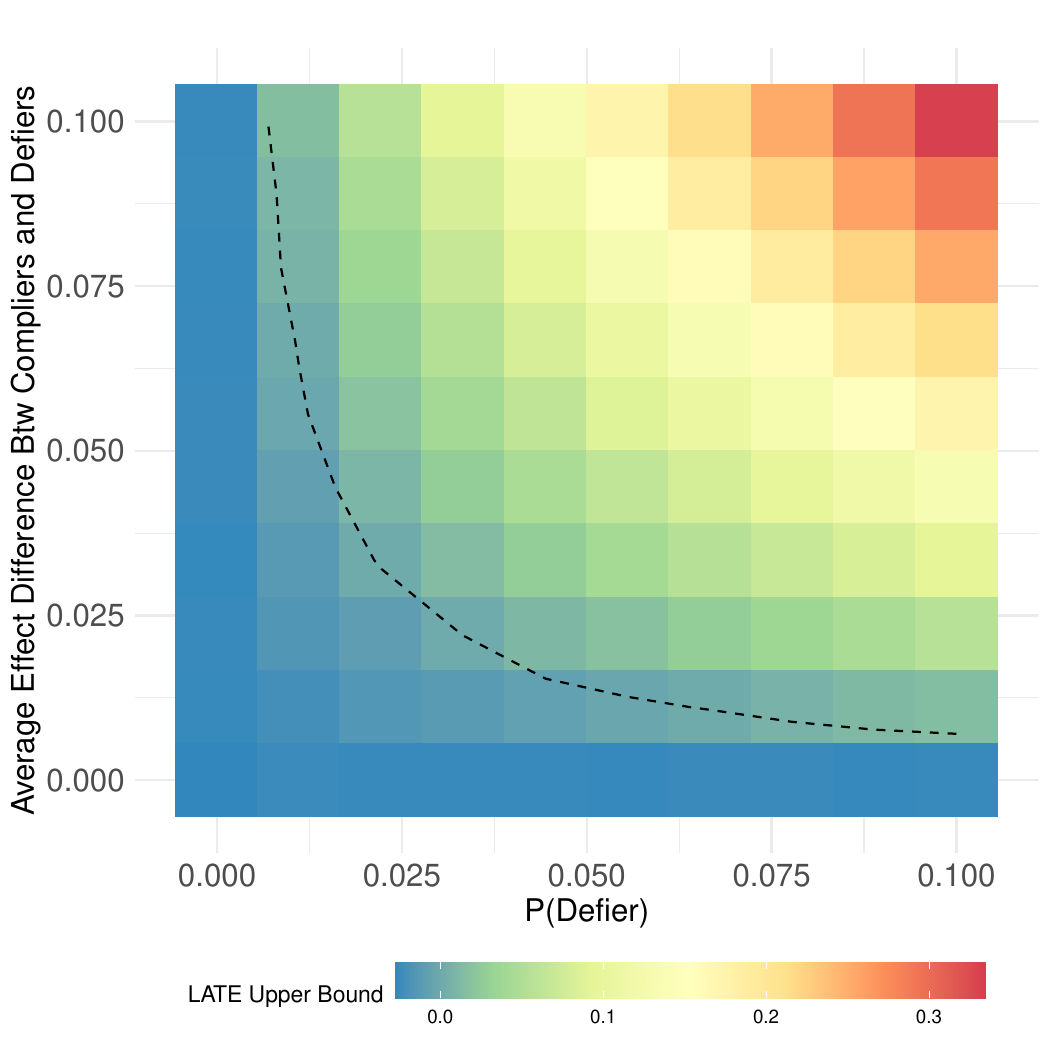}
    \caption{Results for minimum $\delta$ value.}
    \label{fig:s1}
  \end{subfigure}
  \begin{subfigure}[t]{0.45\textwidth}
    \includegraphics[width=\textwidth]{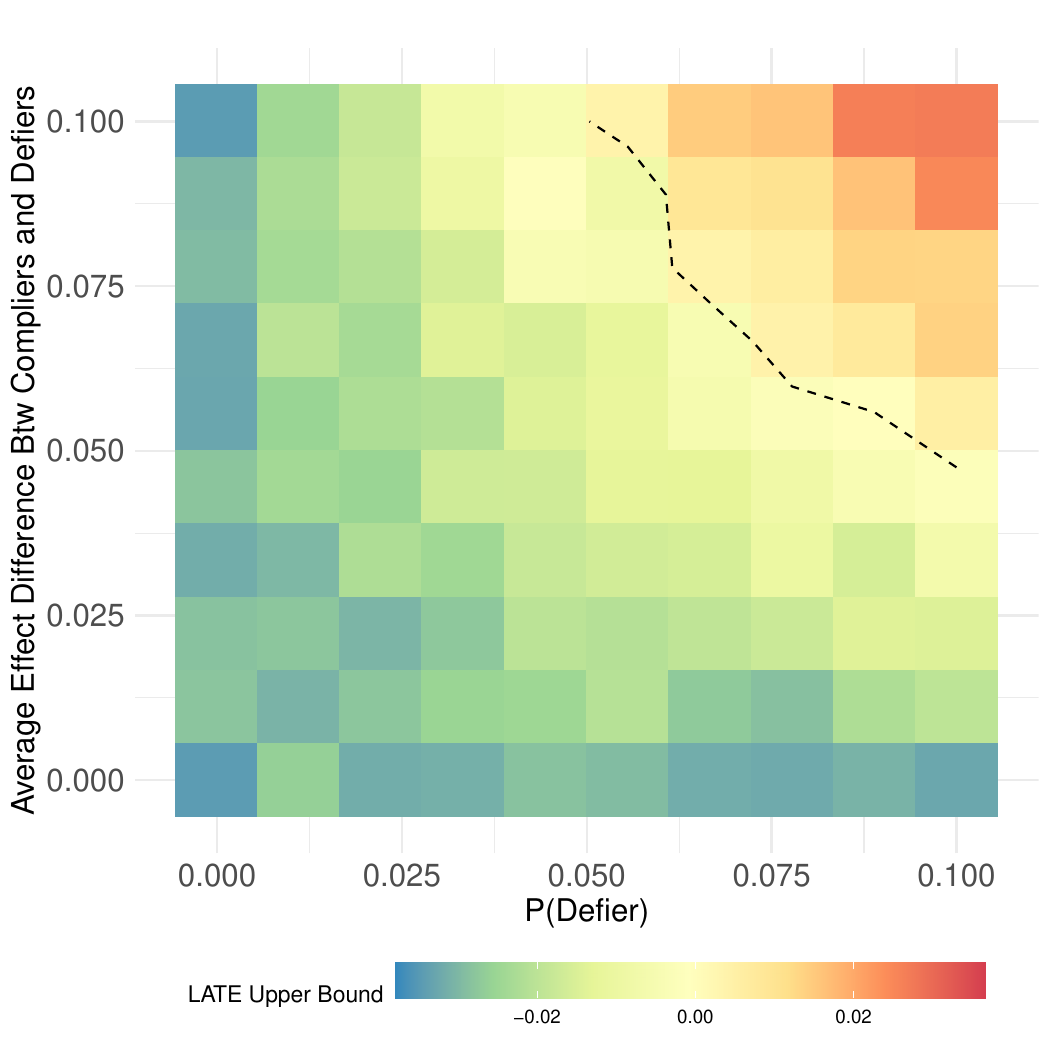}
    \caption{Results for maximum $\delta$ value.}
    \label{fig:s2}
  \end{subfigure}
  \caption{Sensitivity analysis results for the monotonicity assumption. Dotted line represents frontier at which the treatment effect changes sign.}
  \label{fig:mono-sens}
\end{figure}

\section{Discussion}

In this paper, we introduced a method for estimating effects based on dynamic interventions for instrumental variables. Our methods allow researchers to study the impact of modifying each individual's observed instrument value rather than fixing it at a specific level. We demonstrated that this effect can be identified and estimated without relying on a positivity assumption. In addition, we established a general efficiency theory, and we created an influence-function-based estimator that achieves fast convergence rates, even when flexible nonparametric methods are used for estimating nuisance functions. Additionally, we developed methods for uniform inference, profiling principal strata that depend on $\delta$, and a sensitivity analysis for the monotonicity assumption.

It is worth revisiting Proposition~\ref{prop:dominance} in section \ref{sec:proposed}. Proposition~\ref{prop:dominance} both  establishes monotonicity for defining and identifying the $\delta$-tilted Local Average Treatment Effect (LATE) and aids in the interpretation of this estimand. We demonstrated that $\psi(\delta)$ reflects the average treatment effect for individuals who would opt for treatment if the observed instrument were replaced by a ``stochastically'' elevated version, $Z_\delta$, yet remain untreated under the observed instrument $Z$. This implies a higher likelihood of treatment initiation when the instrumental variable is ''stochastically'' elevated by $\delta$, thereby defining the corresponding complier population.

We applied our methods to data on the effectiveness of operative care for patients with cholecystitis. In this application, the IV is a continuous measure of the proportion of times surgeons operate. Our methods allow us to avoid using a binary version of this instrument, which requires arbitrary choices about dichotomization of the IV. In addition, we are able to avoid the positivity assumption, which requires that all patients have an equal chance of being treated by physicions with the same tendency to operate. We also included a profiling analysis that delves into how patient characteristics influence compliance with hypothetical shifts in surgeon behavior and subsequent treatment outcomes. These findings provide nuanced understanding of how patient characteristics interact with treatment decisions and outcomes, informing more targeted and effective healthcare interventions. Finally, we observe that our study is more robust to violations of the monotoncity assumption for larger values of $\delta$.

One avenue for future research would be to establish a minimax lower bound for the tilted LATE parameter. This would allow a more complete investigation of how this parameter depends on $\delta$ under more extreme shifts. Additionally, future research could also consider deriving the concept of instrument sharpness for the tilted LATE parameter.

\section*{References}
\bibliographystyle{asa}
\bibliography{bibliography.bib}

\begin{appendices}
\spacingset{1.5}

\section{Proofs of Main Results}
\label{sec:proofs}

\begin{proof}[Proof of Proposition~\ref{prop:tilt-LATE}]
\noindent Let's consider the case where $Q(\,\cdot \mid X)$ stochastically dominates $\Pi(\,\cdot \mid X)$. The other case can be similarly tackled.

\vspace{2 mm}
Letting $Z_{Q} := Q^{-1}\left(\Pi(Z \mid X) \mid X\right)$, we have
$$
Z_{Q} \geq Z \iff \Pi(Z \mid X) \geq Q(Z \mid X),
$$
 and the right-hand side is satisfied by the assumption of stochastic dominance. Thus by monotonicity, $A^{Z_{Q}} \geq A \equiv A^Z$. Further by construction, $Z_{Q} \mid X \sim Q(\cdot \mid X)$.
\end{proof}

\begin{proof}[Proof of Proposition~\ref{prop:dominance}]
      Assume $\delta > 0$. Observe that $Q_{\delta}(\cdot\mid X)$ stochastically dominates $\Pi(\cdot \mid X) := \Pb(Z \leq z \mid X)$. Specifically,
\begin{equation*}
    \begin{aligned}
        \Pi(z \mid X) - Q_{\delta}(z \mid X) & = \mathbb{P}(Z \leq z \mid X) - \frac{\int_{-\infty}^z e^{\delta u}\pi(u \mid X)du}{\E\left(e^{\delta Z}\mid X\right)} \\
        & = \E\left[I(Z \leq z)\left\{1-\frac{e^{\delta Z}}{\E\left(e^{\delta Z}\mid X\right)}\right\}\mid X\right] \\
        & \geq 0
    \end{aligned}
\end{equation*}
The last line follows since, for $\delta > 0$ 
$$
\E\left(I(Z \leq z)e^{\delta Z} \mid X\right) \leq \Pb(Z \leq z \mid X)\E\left(e^{\delta Z}\mid X\right).
$$
This is due to the following lemma.
\begin{lemma}[Theorem A.19 of \citet{devroye2013}]
    Let $f$ and $g$ be real-valued measurable non-decreasing functions and $X$ be a random variable. Then, $\E(f(X)g(X)) \geq \E(f(X))\E(g(X))$.
    \label{lem:fg}
\end{lemma}
\end{proof}

\begin{proof}[Identification of $\psi(\delta)$]
    Note that, for $\delta > 0$ 
    \begin{equation}
        \mathbb{E}(Y^1 - Y^0 \mid A^{Z_{\delta}}>A) = \mathbb{E}\{Y^1-Y^0 I(A^{Z_{\delta}}>A)\}/\mathbb{P}(A^{Z_{\delta}}>A).
        \label{eq: psi num den}
    \end{equation} 
    Now,
\begin{equation}
    \begin{aligned}
        \mathbb{P}(A^{Z_{\delta}}>A) & = \Pb(A^{Z_{\delta}} = 1, A = 0) \\
        & = \mathbb{E}\{I(A^{Z_{\delta}} = 1, A = 0)\} \\
        & = \E[\E\{I(A^{Z_{\delta}} = 1, A = 0)\mid X\}] \\
        & = \E\{\E(A^{Z_{\delta}} - A \mid X)\}  \\
        & = \E\{\E(A^{Z_{\delta}}\mid X)\} - \E(A) \\
        & = \E(\gamma_{\delta}^A(X)) - \E(A)
    \end{aligned}
    \label{eq: psi den}
\end{equation}
where the fourth equality follows due to monotonicity (Assumption~\ref{ass:mono}) combined with Propositions \ref{prop:tilt-LATE} and \ref{prop:dominance} for $\delta > 0$. Recall that for any random variable $T$ and $\alpha_{\delta}(X) := \E(e^{\delta Z})$, $\gamma_{\delta}^T(X) := \E(Te^{\delta Z})/\alpha_{\delta}(X)$.

Similarly, using Exclusion restriction,
\begin{equation}
    \begin{aligned}
        \E\{(Y^1 - Y^0)I(A^{Z_{\delta}} > A)\} & = \E[\E\{(Y^1 - Y^0)I(A^{Z_{\delta}} >A) \mid X\}] \\
        & = \E(\gamma^Y_{\delta}(X)) - \E(Y)
    \end{aligned}
    \label{eq: psi num}
\end{equation}
Plugging in \eqref{eq: psi den} and \eqref{eq: psi num} in \eqref{eq: psi num den} we have \eqref{eq: tilted LATE}. The case for $\delta < 0$ can be similarly tackled.
\end{proof}

\begin{proof}[Proof of Theorem \ref{thm : eff IF}]
Recall that our parameter of interest is
\begin{equation}
    \psi(\delta) := \frac{\E\left[\gamma_{\delta}^Y(X)\right] - \E(Y)}{\E\left[\gamma_{\delta}^A(X)\right] - \E(A)}
\end{equation}
To find efficient influence function of this quantity, first we find the influence function of each piece:
\begin{enumerate}
    \item[(a)] $\E\left[\gamma_{\delta}^Y(X)\right] = \iint \mathbb{E}(Y \mid Z=u, X=x)  dQ_\delta(u \mid X=x) \mathbb{P}(d x)$, by the definition of $Q_{\delta}(. \mid X)$.
    \item[(b)] $\E(Y)$.
    \item[(c)] $\E\left[\gamma_{\delta}^A(X)\right] = \iint \mathbb{E}(A \mid Z=u, X=x)  dQ_\delta(u \mid X=x) \mathbb{P}(d x)$, by the definition of $Q_{\delta}(. \mid X)$. 
    \item[(d)] $\E(A)$. 
\end{enumerate}
Using the derivative rules while dealing with influence functions we have
\begin{equation}
    IF\left(\frac{f}{g}\right) = \frac{IF(f)g - IF(g)f}{g^2}
    \label{eq: derivative rule}
\end{equation}
Below we give the derivation of the influence function for quantity (a).
\begin{equation*}
        \begin{aligned}
            IF\left(\E(\gamma_{\delta}^Y(X))\right) & = IF\left(\iint \mathbb{E}(Y \mid Z=u, X=x)  dQ_\delta(u \mid X=x) \mathbb{P}(d x)\right) \\
            & = \sum_u \sum_x \frac{I(X = x)I(Z = u)}{\Pb(Z = u \mid X)\Pb(X =x)}(Y - \mathbb{E}(Y \mid Z=u, X=x))dQ_{\delta}(u \mid X = x)\Pb(X = x) \\
            &\hspace{20 pt} + \sum_u\sum_x \E(Y \mid Z = u, X = x)\text{IF}(dQ_{\delta}(u \mid X = x))\Pb(X = x) \\
            &\hspace{20 pt} + \sum_u\sum_x \E(Y \mid Z = u, X = x)dQ_{\delta}(u \mid X = x)\left\{I(X = x) - \Pb(X = x)\right\} \\
            & = \sum_u \sum_x \frac{I(X = x)I(Z = u)}{\Pb(Z = u \mid X)}(Y - \mathbb{E}(Y \mid Z=u, X=x))dQ_{\delta}(u \mid X = x) \\
            &\hspace{20 pt} + \sum_u\sum_x \E(Y \mid Z = u, X = x)\frac{e^{\delta u}I(X = x)I(Z = u)}{\int e^{\delta u}\mathbb{P}(R = u \mid X = x)du} \\
            &\hspace{20 pt} - \sum_u\sum_x \E(Y \mid Z = u, X = x) \frac{\sum_r e^{\delta r}I(Z = u)I(X = x)dQ_{\delta}(u \mid X =x)}{\int e^{\delta u}\mathbb{P}(Z = u \mid X = x)du} \\
            &\hspace{20 pt} + \sum_u\sum_x \E(Y \mid Z = u, X = x)dQ_{\delta}(u \mid X = x)I(X = x) \\
            &\hspace{20 pt} - \sum_u\sum_x \E(Y \mid Z = u, X = x)dQ_{\delta}(u \mid X = x)\Pb(X = x) \\
            &=  \left\{Y -  \mathbb{E}(Y \mid Z, X)\right\}\frac{dQ_{\delta}(Z \mid X)}{\pi(Z \mid X)} + \mathbb{E}(Y \mid Z, X)\frac{e^{\delta Z}}{\int e^{\delta u}\pi(u \mid X)du} \\
            &\hspace{20 pt} + \sum_u \sum_x \E(Y \mid Z = u, X = x)I(X = x)dQ_{\delta}(u \mid X = x)\left\{1-\frac{e^{\delta Z}}{\int e^{\delta u}\pi(u \mid X)du}\right\} - \E[Y^{Z_{\delta}}] \\
            & = \frac{Y e^{\delta Z}}{\alpha_{\delta}(X)} + \left(1-\frac{e^{\delta Z}}{\alpha_{\delta}(X)}\right)\gamma^Y_{\delta}(X) - \E[Y^{Z_{\delta}}]
        \end{aligned}
    \end{equation*}
    where $\alpha_{\delta}(X) := \E(e^{\delta Z}\mid X) = \int e^{\delta u}\pi(u \mid X)du$ and for any random variable $T$, $\gamma^T_{\delta}(X) := \frac{\E(Te^{\delta Z}\mid X)}{\alpha_{\delta}(X)}$.
    Similarly we have,
    \begin{equation}
        IF\left(\E(\gamma_{\delta}^A(X))\right) = \frac{A e^{\delta Z}}{\alpha_{\delta}(X)} + \left(1-\frac{e^{\delta Z}}{\alpha_{\delta}(X)}\right)\gamma^A_{\delta}(X) - \E[A^{Z_{\delta}}]
    \end{equation}
    For any random variable $T$, define $\xi(T; \delta) := \frac{Te^{\delta Z}}{\alpha_{\delta}(X)} + \gamma_{\delta}^T(X)\left(1 - \frac{e^{\delta Z}}{\alpha_{\delta}(X)}\right)$ and further take $\Xi(T; \delta, 0) := \xi(T; \delta) - \xi(T; 0)$.
    Finally it is straightforawrd to show that
    \begin{equation}
        IF(\E(Y)) = Y - \E(Y); \;\; IF(\E(A)) = A - \E(A)
    \end{equation}
    Plugging these expressions into \eqref{eq: derivative rule},
    \begin{equation}
        \begin{aligned}
            \varphi_{\delta}(z; \eta, \psi) & = \frac{g\{\Xi(Y; \delta, 0)-f\} -f\{\Xi(A;\delta, 0)-g\}}{g^2} \\
            & = \frac{\Xi(Y; \delta, 0) - \psi\Xi(A; \delta, 0)}{g} \quad \text{since } \psi = \frac{f}{g} \\
            & = \frac{\Xi(Y; \delta, 0) - \psi\Xi(A; \delta, 0)}{\E\left(A^{Z_{\delta}} - A\right)}
        \end{aligned}
    \end{equation}
\end{proof}

\begin{proof}[Proof of Theorem \ref{thm: asymptotic normality}]
Let us re-define $f(Y;\eta) := \Xi(Y;\delta,0)$ and $g(A;\eta) := \Xi(A;\delta,0)$ where $\eta$ denotes the vector of nuisance functions $\left(\alpha_{\delta}(X),\gamma_{\delta}^Y(X),\gamma_{\delta}^A(X)\right)$.
    We then have $\widehat{\psi}_{IF}(\delta) = \frac{\Pn f(Y; \widehat{\eta})}{\Pn g(A; \widehat{\eta})}$ so that
    \begin{equation}
    \widehat{\psi}_{IF}(\delta) - \psi(\delta) = \frac{1}{\Pn g(A;\widehat{\eta
    })}\left[\Pn f(Y;\widehat{\eta})-\Pb f(Y;\eta) - \psi\{\Pn g(A;\widehat{\eta})-\Pb g(A;\eta)\}\right]
    \label{eq: decomposition}
    \end{equation}
    We can treat the first ratio as a constant under the causal assumptions and break down the other components into smaller pieces, for example:
    \begin{equation}
        \Pn f(Y;\widehat{\eta})-\Pb f(Y;\eta) = (\Pn - \Pb)\{f(Y;\widehat{\eta})-f(Y; \eta)\} + (\Pn - \Pb)f(Y;\eta) - P\{f(Y;\widehat{\eta})-f(Y; \eta)\}
    \end{equation}
    The first term on the RHS is $o_{\Pb}(1/\sqrt{n})$ as long as we estimate the nuisance parameters on a separate sample and assume $\|1-\frac{\alpha_{\delta}(X)}{\widehat{\alpha}_{\delta}(X)}\|+\|\widehat{\gamma}_{\delta}^Y(X) - \gamma_{\delta}^Y(X)\|+\|\widehat{\gamma}_{\delta}^A(X) - \gamma_{\delta}^A(X)\| = o_{\Pb}(1)$ \citep{kennedy2020b}. The second term is asymptotically Gaussian by CLT. The third term can be decomposed as
    \begin{equation}
        \begin{aligned}
            \Pb f(Y; \widehat{\eta}) - \Pb f(Y; \eta) & = \E\left[\frac{e^{\delta Z}Y}{\widehat{\alpha}_{\delta}(X)} + \widehat{\gamma}_{\delta}^Y(X)\left(1-\frac{e^{\delta}Z}{\widehat{\alpha}_{\delta}(X)}\right) - \frac{e^{\delta Z}Y}{\alpha_{\delta}(X)}\right] \\
            & = \E\left[\left(\widehat{\gamma}_{\delta}^Y(X) - \gamma_{\delta}^Y(X)\right)\left(1-\frac{\alpha_{\delta}(X)}{\widehat{\alpha}_{\delta}(X)}\right)\right]\\
            & \lesssim \left\|\widehat{\gamma}_{\delta}^Y(X) - \gamma_{\delta}^Y(X)\right\|\left\|1-\frac{\alpha_{\delta}(X)}{\widehat{\alpha}_{\delta}(X)}\right\|
        \end{aligned}
        \label{eq: f guarantee}
    \end{equation}
    Similarly we have 
    \begin{equation}
        \Pb g(A; \widehat{\eta}) - \Pb g(A; \eta) = \lesssim \left\|\widehat{\gamma}_{\delta}^A(X) - \gamma_{\delta}^A(X)\right\|\left\|1-\frac{\alpha_{\delta}(X)}{\widehat{\alpha}_{\delta}(X)}\right\|
        \label{eq: g guarantee}
    \end{equation}
    Finally plugging in \eqref{eq: f guarantee} and \eqref{eq: g guarantee} in \eqref{eq: decomposition}, we have \eqref{eq: convergence}.
    To achieve consistent estimation of $\psi(\delta)$, we therefore need
    $$
    \left\|1-\frac{\alpha_{\delta}(X)}{\widehat{\alpha}_{\delta}(X)}\right\|\left\{\left\|\widehat{\gamma}_{\delta}^Y(X)-\gamma_{\delta}^Y(X)\right\|+ \left\|\widehat{\gamma}_{\delta}^A(X)-\gamma_{\delta}^A(X)\right\|\right\} = o_{\Pb}(1/\sqrt{n})
    $$
    This is satisfied if all terms are estimated at faster than $n^{-1/4}$ or if $\alpha_{\delta}(.)$ (or both $\gamma_{\delta}^Y(.)$ and $\gamma_{\delta}^A(.)$) are estimated at $n^{-1/2}$.
\end{proof}

\begin{proof}[Proof of Theorem \ref{thm: uniform normality}]
    It follows the proof of Theorem 3 in \cite{mauro2020instrumental}.
\end{proof}

\section{Profiling}
\label{sec: app profiling}

\begin{proof}[Profiling]
Let $V^{\prime}$ be $X\backslash V$ such that $X = (V, V^{\prime})$. Furthermore, we assume that $V^{\prime}$ follows the marginal density $d\widetilde{Q}$.

\textbf{Complier :}

For $\delta > 0$
\begin{equation}
    \begin{aligned}
        \Pb(A^{Z_{\delta}} = 1, A = 0 \mid V = v) & = \int \E\left(A^{Z_{\delta}} - A \mid X = x\right)d\widetilde{Q}(v^{\prime}) \\
        & = \int \left(\gamma_{\delta}^A(x) - \E\left(A \mid X = x\right)\right)d\widetilde{Q}(v^{\prime})
    \end{aligned}
\end{equation}
The second equality follows from clubbing assumption 4 along with propositions \ref{prop:dominance} and \ref{prop:tilt-LATE}.
Therefore, by Bayes' theorem
\begin{equation}
\begin{aligned}
    \Pb\left(V = v_0 \mid A^{Z_{\delta}} = 1, A = 0\right) & = \int 1(v = v_0)\Pb\left(V = v \mid A^{Z_{\delta}} = 1, A = 0\right)dQ(v) \\
    & = \frac{\int 1(V = v_0)\Pb\left(A^{Z_{\delta}} = 1, A = 0 \mid V = v\right)dQ(v)}{\Pb\left(A^{Z_{\delta}} = 1, A = 0\right)} \\
    & = \frac{\E\left[1(V = v_0)\left\{\gamma_{\delta}^A(X) - \E(A \mid X)\right\}\right]}{\E(\gamma_{\delta}^A(X) - \E(A \mid X))}
\end{aligned}
\end{equation}
Similarly, the case for $\delta < 0$ can be shown to have the same expression.

\textbf{Always Taker :}
 For $\delta > 0$,
\begin{equation}
    \begin{aligned}
        \Pb(A^{Z_{\delta}} = A = 1 \mid V = v) & = \int \E\left(A^{Z_{\delta}}A \mid X = x\right)d\widetilde{Q}(v^{\prime}) \\
        & = \int \E\left(A \mid X = x\right)d\widetilde{Q}(v^{\prime}) \\
        & = \E(A \mid V = v)
    \end{aligned}
\end{equation}
The second equality follows from clubbing assumption 4 along with propositions \ref{prop:dominance} and \ref{prop:tilt-LATE}.
Therefore 
\begin{equation}
\begin{aligned}
    \Pb\left(V = v_0 \mid A^{Z_{\delta}} = A = 1\right) & = \int 1(v = v_0)\Pb\left(V = v \mid A^{Z_{\delta}} = A = 1\right)dQ(v) \\
    & = \frac{\int 1(V = v_0)\Pb\left(A^{Z_{\delta}} = A = 1 \mid V = v\right)dQ(v)}{\Pb\left(A^{Z_{\delta}} = A = 1\right)} \\
    & = \frac{\E\left(1(V = v_0)\E(A \mid X)\right)}{\E(A)}
\end{aligned}
\end{equation}
Now, for $\delta < 0$,
\begin{equation}
    \begin{aligned}
        \Pb(A^{Z_{\delta}} = A = 1 \mid V = v) & = \int \E\left(A^{Z_{\delta}}A \mid X = x\right)d\widetilde{Q}(v^{\prime}) \\
        & = \int \E\left(A^{Z_{\delta}} \mid X = x\right)d\widetilde{Q}(v^{\prime}) \\
        & = \int \gamma_{\delta}^A(x)d\widetilde{Q}(v^{\prime})
    \end{aligned}
\end{equation}
Therefore, 
\begin{equation}
\Pb\left(V = v_0 \mid A^{Z_{\delta}} = A = 1\right)  = \frac{\E\left(1(V = v_0)\gamma_{\delta}^A(X)\right)}{\E\left(\gamma_{\delta}^A(X)\right)}
\end{equation}

\textbf{Never Taker :}
\begin{equation}
    \begin{aligned}
        \Pb(A^{Z_{\delta}} = A = 0 \mid V = v) & = \int \E\left((1 - A^{Z_{\delta}}A) \mid X = x\right)d\widetilde{Q}(v^{\prime}) \\
        & = \int \E\left(1 - A^{Z_{\delta}} \mid X = x\right)d\widetilde{Q}(v^{\prime}) \\
        & = \int \left(1-\gamma_{\delta}^A(x)\right)d\widetilde{Q}(v^{\prime}) \quad Z \perp A(z) \forall z
    \end{aligned}
\end{equation}
where the second equality follows by combining assumption 4 with propositions \ref{prop:dominance} and \ref{prop:tilt-LATE}.
Therefore, by Bayes' theorem,
\begin{equation}
\Pb\left(V = v_0 \mid A^{Z_{\delta}} = A = 0\right) = \frac{\E\left(1(V = v_0)\left(1- \gamma_{\delta}^A(X)\right)\right)}{\E\left(1- \gamma_{\delta}^A(X)\right)}    
\end{equation}
Now, for $\delta < 0$,
\begin{equation}
    \begin{aligned}
        \Pb(A^{Z_{\delta}} = A = 0 \mid V = v) & = \int \E\left(1 - A^{Z_{\delta}}A \mid X = x\right)d\widetilde{Q}(v^{\prime}) \\
        & = \int \E\left(1 - A \mid X = x\right)d\widetilde{Q}(v^{\prime}) \\
        & = \E(1 - A \mid V = v)
    \end{aligned}
\end{equation}
Therefore, 
\begin{equation}
\Pb\left(V = v_0 \mid A^{Z_{\delta}} = A = 1\right)  = \frac{\E\left(1(V = v_0)\E(1 - A \mid X)\right)}{\E(1 - A)}
\end{equation}
\end{proof}

Under the assumptions in Section~\ref{sec:assump}, the density for baseline covariate at $V = v_0$ among the compliers ( for $\delta < 0$) is identified as follows:
\begin{equation}
    \Pb(V = v_0 \mid A^{Z_{\delta}} < A) = \frac{\E[D_{\psi}I(V = v_0)]}{\E[D_{\psi}]},
    \label{eq: comp profiling dl0}
\end{equation}
where $D_{\psi} = \gamma_{\delta}^A(X) - \E(A \mid X)$. 

An expression similar to \eqref{eq: comp profiling dl0} can be derived for never-takers ($A^{Q_{\delta}} = 0 = A$):
\begin{equation}
\Pb\left(V = v_0 \mid A^{Z_{\delta}} = A = 0\right) = \frac{\E\left(1(V = v_0)\E(1-A \mid X)\right)}{\E(1-A)}.
\label{eq: nt profiling dl0}
\end{equation}
Next, the expression for always-takers ($A^{Z_{\delta}} = 1 = A$) is
\begin{equation}
\Pb\left(V = v_0 \mid A^{Z_{\delta}} = A = 1\right) = \frac{\E\left(1(V = v_0)\gamma_{\delta}^A(X)\right)}{\E\left(\gamma_{\delta}^A(X)\right)}.
\label{eq: at profiling dl0}
\end{equation}
Notably, contrary to the case of $\delta >0$, here the expression for never-takers does not depend on $\delta$. The rest proceeds similarly to the case of $\delta > 0$ in Section \ref{sec: profiling}.

\section{Restrictions on the Defier Population}
\label{sec: app defier}

\citet{richardson2010analysis} showed that the observed distribution of compliers constrains the proportion of defiers
when the instrument is a binary variable. Here, we demonstrate that a similar form of constraint exists under our estimand
with a continuous instrument. First, let $t_X$ to denote generic compliance type in the set $D_X = \{AT, NT, CO, DE\}$, where $AT = \text{Always Taker}$, $NT = \text{Never Taker}$, $CO = \text{Complier}$, $DE = \text{Defier}$. Let $\pi_{t_X} \equiv p(t_X)$ denote the marginal probability of a given compliance type $t_X \in D_X$, and let $\pi_X \equiv {\pi_{t_X} \mid t_X \in D_X}$ denote a marginal distribution on $D_X$.

\noindent {\bf Case 1:}
For every $z > z^{\prime}$,
\begin{equation*}
    \begin{aligned}
        \PP(A = 1 \mid Z = z^{\prime}) & = \PP(A^z = 0, A^{z^{\prime}} = 1) + \PP(A^z = 1, A^{z^{\prime}} = 1) \\
        = \PP(\text{DE}) + \PP(\text{AT})
    \end{aligned}
\end{equation*}
Similarly,
\begin{equation*}
\PP(A = 1 \mid Z = z) = \PP(\text{CO}) + \PP(\text{AT})
\end{equation*}
Define $c_1 = \PP(A = 1 \mid Z = z^{\prime}), c_2 = \PP(A = 1 \mid Z = z)$.
$$
P_{c_1, c_2} = \left\{\pi_{AT} = t, \pi_{DE} = c_1 - t, \pi_{co} = c_2 - t, \pi_{NT} = 1 - c_1 - c_2 + t, \; \; t \in \left[\max\{0, c_1+c_2 -1\}, \min\{c_1, c_2\}\right]\right\}
$$
where $P_{c_1,c_2}$ is the set of joint distributions $\pi_X$ on $D_X$ that are compatible with $\Pb(x | z)$.

\noindent {\bf Case 2:}
For $\delta > 0$, let $Z_{\delta} \sim Q_{\delta}, Z \sim \Pi$. It can then be proved,
\begin{equation*}
    \begin{aligned}
        \PP(A = 1 \mid Z) & = \PP(A^{Z_{\delta}} = 0, A^Z = 1) + \PP(A^{Z_{\delta}} = 1, A^Z = 1) \\
        & = \PP(\text{DE}_{\delta}) + \PP(\text{AT}_{\delta})
    \end{aligned}    
\end{equation*}
and 
\begin{equation*}
\PP(A = 1 \mid Z_{\delta}) = \PP(\text{CO}_{\delta}) + \PP(\text{AT}_{\delta})
\end{equation*}
Define $U_1 := A^{Z_{\delta}}, U_2 := A^Z$ and $c_1 = \PP(A = 1 \mid Z), c_2 = \PP(A = 1 \mid Z_{\delta})$
$$
P_{c_1, c_2} = \left\{\pi_{AT} = t, \pi_{DE} = c_1 - t, \pi_{co} = c_2 - t, \pi_{NT} = 1 - c_1 - c_2 + t, \; \; t \in \left[\max\{0, c_1+c_2 -1\}, \min\{c_1, c_2\}\right]\right\}
$$
where $P_{c_1, c_2}$ is defined as in the previous case.

\section{Sensitivity Analysis}
\label{sec: app sensitivity}
For $\delta < 0$ we have
\begin{equation}
    \E\left(Y^1 - Y^0 \mid A^{Z_{\delta}} < A\right) = \psi(\delta) + \frac{\gamma_1(\delta)\gamma_2(\delta)}{\E\left( A - A^{Z_{\delta}}\right)}
\end{equation}
where for a pre-specified $\delta, \gamma_1(\delta)$ signifies the proportion of defiers, while $\gamma_2(\delta)$ represents the discrepancy in average treatment effects between defiers and compliers. Specifically, $\gamma_1(\delta)$ and $\gamma_2(\delta)$ are defined as
\begin{equation*}
    \begin{aligned}
        \gamma_1(\delta) & := \Pb\left(A^{Z_{\delta}} > A\right) \\
        \gamma_2(\delta) & := \E\left(Y^1 - Y^0 \mid A^{Z_{\delta}} > A\right) - \E\left(Y^1 - Y^0 \mid A^{Z_{\delta}} < A\right)
    \end{aligned}
\end{equation*}
The rest proceeds similarly to the case of $\delta > 0$ in Section \ref{sec: sensitivity}.

\section{Identification of $\psi(\delta_1, \delta_2) = \E\left(Y^1 - Y^0 \mid A^{Z_{\delta_1}} > A^{Z_{\delta_2}}\right)$}
\label{app:d1d2}

One can observe that Proposition~\ref{prop:tilt-LATE} can be generalized, for arbitrary distributions $Q_1, Q_2$, as follows.
\begin{proposition}
    Under Assumption \ref{ass:mono}, if $Q_1$ stochastically dominates $Q_2$, then $\Pb[A^{Z_{Q_{1}}} \geq A^{Z_{Q_{2}}}] = 1$. Similarly, if $Q_{1}$ is stochastically dominated by $Q_{2}$, then $\Pb[A^{Z_{Q_{1}}} \leq A^{Z_{Q_{2}}}] = 1$. 
\label{prop:tilt-LATE general}
\end{proposition}

\begin{proof}
    Without loss of generality assume $Q_1$ stochastically dominates $Q_2$. Recall that by definition, $Z_{Q_j} := Q_j^{-1}(\Pi(Z \mid X) \mid X)$, for $j \in \{1, 2\}$. Now,
    $$
    Z_{Q_1} \geq Z_{Q_2} \iff Q_1^{-1}(\Pi(Z \mid X) \mid X) \geq Q_2^{-1}(\Pi(Z \mid X) \mid X),
    $$
    but the latter holds by stochastic dominance---see Theorem 1' of \citet{levy1978}.
By Assumption \ref{ass:mono}, we conclude that $A^{Z_{Q_1}} \geq A^{Z_{Q_2}}$.
\end{proof}


\noindent Similarly, for arbitrary $\delta_1, \delta_2$, Proposition \ref{prop:dominance} can be generalized as follows.
\begin{proposition}
    For $\delta_1 > \delta_2$, $Q_{\delta_1}$ stochastically dominates $Q_{\delta_2}$; for $\delta_1 < \delta_2, Q_{\delta_1}$ is stochastically dominated by $Q_{\delta_2}$.
\label{prop:dominance general}
\end{proposition}

\begin{proof}
    Assume $\delta_1 > \delta_2 > 0$. 
\begin{equation*}
    \begin{aligned}
        Q_{\delta_1}(z \mid X) - Q_{\delta_2}(z \mid X) & = \E\left[I(Z \leq z)\left\{\frac{e^{\delta_1 Z}}{\E\left(e^{\delta_1 Z}\mid X\right)}-\frac{e^{\delta_2 Z}}{\E\left(e^{\delta_2 Z}\mid X\right)}\right\}\mid X\right] \\
        & \leq 0
    \end{aligned}
\end{equation*}
This is due to
\begin{equation}
\begin{aligned}
    \E\left[I(Z \leq z)e^{\delta_1 Z} \mid X\right] & = -\E\left[-I(Z \leq z)e^{\delta_2 Z}e^{(\delta_1 - \delta_2)Z}\mid X\right] \\
    & \leq -\E\left[-I(Z \leq z)e^{\delta_2 Z}\mid X\right]\E\left[e^{(\delta_1 - \delta_2)Z} \mid X\right] \\
    & \leq \E\left[I(Z \leq z)e^{\delta_2 Z} \mid X\right]\E\left[e^{\delta_1 Z}\mid X\right]\frac{1}{\E\left[e^{\delta_2 Z}\mid X\right]} 
\end{aligned}
\label{eq: ineq d1d2}
\end{equation}
so that,
$$
\frac{\E\left[I(Z \leq z)e^{\delta_1 Z} \mid X\right]}{E\left[e^{\delta_1 Z} \mid X\right]} \leq \frac{\E\left[I(Z \leq z)e^{\delta_2 Z} \mid X\right]}{E\left[e^{\delta_2 Z} \mid X\right]}
$$
The last inequality in \eqref{eq: ineq d1d2} follows since by Lemma \ref{lem:fg},
\begin{equation*}
\begin{aligned}
    \E\left[e^{\delta_1 Z}\mid X\right] & = \E\left[e^{(\delta_1 - \delta_2)Z}e^{\delta_2 Z}\mid X\right] \\
    & \geq \E\left[e^{(\delta_1-\delta_2)Z}\mid X\right]\E\left[e^{\delta_2 Z}\mid X\right].
\end{aligned}    
\end{equation*}
Ultimately, we have for every $z$,
\begin{equation}
    Q_{\delta_1}(z \mid X) \leq Q_{\delta_2}(z \mid X) \leq 0
\end{equation}
Other cases such as $0 > \delta_1 > \delta_2 $ and $\delta_1 < \delta_2$ can be similarly handled.
\end{proof}



\noindent Now let w.l.o.g. $\delta_1 > \delta_2$. Then,
\begin{equation*}
\begin{aligned}
    \Pb(A^{Z_{\delta_1}} > A^{Z_{\delta_2}}) & = \Pb(A^{Z_{\delta_1}} = 1, A^{Z_{\delta_2}} = 0) \\
    & = \E\left(\E(A^{Z_{\delta_1}} - A^{Z_{\delta_2}} \mid X)\right) \\
    & = \E\left\{\gamma_{\delta_1}^A(X) - \gamma_{\delta_2}^A(X)\right\}.
\end{aligned}
\end{equation*} 
The second equality follows from Assumption \ref{ass:mono} clubbed together with propositions \ref{prop:dominance general} and \ref{prop:tilt-LATE general} for $\delta_1 > 0 > \delta_2$.
Similarly,
\begin{equation}
    \E\left\{(Y^1 - Y^0)I(A^{\delta_1} > A^{\delta_1})\right\} = \E\left\{\gamma_{\delta_1}^Y(X) - \gamma_{\delta_2}^Y(X)\right\}
\end{equation}
so that
\begin{equation}
    \psi(\delta_1, \delta_2) = \frac{\E\left\{\gamma_{\delta_1}^Y(X) - \gamma_{\delta_2}^Y(X)\right\}}{\E\left\{\gamma_{\delta_1}^A(X) - \gamma_{\delta_2}^A(X)\right\}}
\end{equation}


\end{appendices}

\end{document}